\documentclass[12pt]{article}

\usepackage{amssymb}
\usepackage{amsmath}

\usepackage{amsmath, amssymb, bm}
\usepackage{amsthm}
\usepackage{xcolor}
\usepackage[letterpaper,margin=1in]{geometry}
\usepackage{mathrsfs}
\usepackage{longtable}
\usepackage{acro} 
\usepackage{caption}
\usepackage{subcaption}
\usepackage{bbm}

\usepackage{hyperref}
\usepackage{algorithmic}
\usepackage[noresetcount,lined,boxed]{algorithm2e} 
\usepackage{tikz}

\usepackage{bbm}

\usepackage{pdflscape}
\usepackage{afterpage}
\usepackage{authblk}

\usetikzlibrary{positioning,fit,backgrounds, calc}

\newtheorem{remark}{Remark}
\newtheorem{example}{Example}
\newtheorem{proposition}{Proposition}

\newtheorem{corollary}{Corollary}

\DeclareMathOperator*{\esssup}{ess\,sup}

\renewcommand{\u}{u}                
\newcommand{\U}{\mathcal{U}}        

\newcommand{\parameterspace}[0]{\R^{\dimParameter}} 
\newcommand{\pdesolutionoperator}[0]{\mathcal{F}}

\newcommand{\Ltwotime}{L^2([0, \finalt])}

\newcommand{\dps}{\displaystyle}
\newcommand{\R}{\mathbb{R}}        
\newcommand{\x}{\bm{x}}            
\newcommand{\finalt}{T}            
\newcommand{\m}{\bm{m}}            
\newcommand{\mz}{\bm{m}_0}         
\newcommand{\mpr}{\bm{m}_\mathrm{pr}}         
\newcommand{\prior}{\mu_{\mathrm{pr}}}
\newcommand{\ObsPath}{\bm{p}}                   
\newcommand{\DAEPath}{\bm{p}}                   
\newcommand{\PathSpace}{\mathcal{P}}            
\renewcommand{\d}{d}               
\newcommand{\dobs}{\d_{\textnormal{obs}}}       

\newcommand{\mini}{\mathop{\mbox{minimize}}}

\newcommand{\MeasKern}{\Phi}       
\newcommand{\Domain}{\Omega}       
\newcommand{\dimDomain}{n_{\Domain}}
\newcommand{\CovPost}{\Sigma_{\rm{post}}(\ObsPath)}     
\newcommand{\CovPostDiscrete}{\Sigma_{\rm{post}}^h(\ObsPath^h)}     
\newcommand{\CovPr}{\Sigma_\mathrm{pr}}         
\newcommand{\CovNoise}{\Sigma_\mathrm{noise}}   
\newcommand{\iCovPr}{\Sigma_\mathrm{pr}^{-1}}           
\newcommand{\OpCovNoise}{\mathcal{C}_{\rm{noise}}}    
\newcommand{\Fisher}{\mathbf{M}_{\rm{Fisher}}(\ObsPath)}
\newcommand{\FisherDiscrete}{\mathbf{M}^h_{\rm{Fisher}}(\ObsPath^h)}
\newcommand{\ObjOED}{\Psi}          
\newcommand{\ObjAOED}{\Psi_{\rm{A}}}
\newcommand{\ObjDOED}{\Psi_{\rm{D}}}

\DeclareMathOperator{\tr}{tr}       
\newcommand{\postMeasure}{\mu_\mathrm{post}}    
\newcommand{\paratoobs}{\mathcal{G}}    
\newcommand{\paratoobsDiscrete}{\mathbf{G}({\ObsPath^h})}    
\newcommand{\paratoobsDiscreteh}{\mathcal{G}^h}

\newcommand{\admissible}{A_{\text{ad}}}
\newcommand{\admissibleDomain}{\Omega_{\text{ad}}}
\newcommand{\regularization}{{R}}

\newcommand{\dimTime}{n_{\rm{t}}}
\newcommand{\dimParameter}{M}

\newcommand{\tk}[1][k]{t_{#1}}
\newcommand{\dt}{\Delta t}
\newcommand{\mbasis}[1][m]{\mathbf{e}_{#1}}
\newcommand{\febasis}[1][k]{\psi_{#1}}
\newcommand{\pdesolutionoperatorDiscrete}[0]{\mathcal{F}^h}
\newcommand{\UDiscrete}[0]{\mathcal{U}^h}
\newcommand{\uDiscrete}[0]{\u^h}
\newcommand{\pk}[1][k]{\mathbf{p}_{#1}}
\newcommand{\pki}[1][k,i]{\mathbf{p}_{#1}}

\newcommand{\Identity}[1][??]{\mathbf{I}_{#1}}

\newcommand{\likelihood}{\pi_{\rm{like}}}

\newcommand{\mpost}{\m_{\ObsPath}}

\newcommand{\velocity}{v}
\newcommand{\angular}{\omega}
\newcommand{\heading}{\theta}

\DeclareAcronym{dae}{
short = \myred{DAE},
long = \myred{differential algebraic equation},
tag = abbrev
}

\DeclareAcronym{ipopt}{
short = IPOPT,
long = Interior Point OPTimizer,
tag = abbrev
}

\DeclareAcronym{nlp}{
short = NLP,
long = nonlinear programming,
tag = abbrev
}

\DeclareAcronym{ode}{
short = ODE,
long = ordinary differential equation,
tag = abbrev
}

\DeclareAcronym{oed}{
short = OED,
long = optimal experimental design,
tag = abbrev
}

\DeclareAcronym{pde}{
short = PDE,
long = partial differential equation,
tag = abbrev,
long-plural = s,
short-plural = s
}

\DeclareAcronym{spd}{
short = s.p.d.,
long = symmetric positive definite,
tag = abbrev
}

\DeclareAcronym{svd}{
short = SVD,
long = singular value decomposition,
tag = abbrev
}

\DeclareAcronym{uq}{
short = UQ,
long = uncertainty quantification,
tag = abbrev
}

\DeclareAcronym{wlog}{
short = w.l.o.g.,
long = without loss of generality,
tag = abbrev
}

\begin{document}

\title{Optimal Experimental Design of a Moving Sensor for Linear Bayesian Inverse Problems}

\author[1]{Nicole Aretz}
\author[2]{Thomas Lynn}
\author[1]{Karen Willcox}
\author[3]{Sven Leyffer}

\affil[1]{Oden Institute for Computational Engineering and Sciences, \newline University of Texas at Austin}
\affil[2]{Johns Hopkins Applied Physics Laboratory}
\affil[3]{Argonne National Laboratory}

\date{}

\maketitle

\begin{abstract}
We optimize the path of a mobile sensor to minimize the posterior uncertainty of a Bayesian inverse problem.
Along its path, the sensor continuously takes measurements of the state, which is a physical quantity modeled as the solution of a partial differential equation (PDE) with uncertain parameters.
Considering linear PDEs specifically, we derive the closed-form expression of the posterior covariance matrix of the model parameters as a function of the path, and formulate the optimal experimental design problem for minimizing the posterior's uncertainty.
We discretize the problem such that the cost function remains consistent under temporal refinement.
Additional constraints ensure that the path avoids obstacles and remains physically interpretable through a control parameterization.
The constrained optimization problem is solved using an interior-point method.
We present computational results for a convection-diffusion equation with unknown initial condition.
\end{abstract}

\section{Introduction}\label{sec:intro}

We optimize the path of a mobile sensor that takes measurements of a physical state for the purpose of inferring parameters in the state's governing dynamics.
The parameters are only indirectly observable through their influence on the state, described by a linear \ac{pde}.
Given measurements, Bayes theorem can be used to update prior information on the parameters; however, the uncertainty in the posterior distribution depends on the path along which the measurements are taken.
When resources are constrained (e.g., because measurements by a mobile sensor may be expensive, labor-intensive, or need to happen in time-critical scenarios) preparatory \ac{oed} is essential to reduce the posterior uncertainty as much as possible.

The optimal steering of moving devices is well-explored by the optimal control community for various purposes, including the steering of mobile sensors for parameter and system identification.
In waypoint-based approaches, the sensor moves between waypoints, with the exact path determined, for example, a priori through a network search (\cite{hitz2014fully, binney2013optimizing, binney2010informative}) or dynamically in real-time (\cite{denniston2019comparison, euler2014centralized, peng2015dddas}).
When parameters are only indirectly observable (e.g., when identifying a plume's original origin rather than monitoring its spread), an inverse problem needs to be solved to update uncertainties from measurement data.
Inverse problems are computationally expensive to solve and thus require reduced-order modeling or other emulators to enable real-time steering (\cite{park2010learning, wogrin2023data}) or more general sequential design policies (\cite{huan2016sequential, shen2025variational, shen2023bayesian}).
In sensor selection for \ac{pde}-constrained inverse problems, stationary sensor positions are selected to reduce the uncertainty of the inferred parameter.
The combinatorial sensor selection problem is solved, for example, through convex relaxation (\cite{alexanderian2014optimal,attia2018goal, attia2022optimal}), greedy algorithms (\cite{binev2018greedy, wu2023fast, aretz2024greedy, maday2015parameterized}), or combinatorial integer approximation approaches (\cite{sager2011combinatorial,sager2012integer, kirches2021compactness}).
The global optimization of continuous sensor paths for mobile sensors has, for example, been explored in \cite{song2007optimal, ucinski2005time, ucinski1999path, Tricaud2008} for deterministic inverse problems, in \cite{ha2019periodic, demetriou2011state, demetriou2009estimation} for state estimation via filtering, and in \cite{choi2010continuous, choi2009adaptive} for informative forecasting.
We refer to \cite{huan2024optimal, ucinski2004optimal} for introductions to \ac{oed}.

In this work, we formulate an \ac{oed} problem over the continuous path of a mobile sensor to reduce the uncertainty in the solution of a \ac{pde}-constrained Bayesian inverse problem that is informed by continuous or frequent measurements along the path.
An overview of this nested structure is provided in Figure \ref{fig:overview}.
Considering linear \acp{pde} specifically, we analyze the connection between the uncertainty in the posterior distribution and the sensor path, providing in particular gradients for use in the optimization.
Nonlinear \acp{pde} may be considered after a Laplace approximation.
We adopt a control-based characterization of the path as the solution of an \ac{ode} which guarantees that the obtained path is realistic.
This distinction between the path and its control, moreover, allows for a straightforward integration of both path and control constraints, e.g., obstacles or minimum turn radius, such that the \ac{oed} problem can be solved efficiently through an interior point algorithm even for non-convex domains.
We show how to discretize the \ac{oed} problem consistently such that it converges under arbitrarily fine time discretization. 
As a consequence, the optimization algorithm may be warm-started with optimal paths obtained for coarser resolutions to save computational resources.
We demonstrate our \ac{oed} formulation and algorithm on a convection-diffusion model for the spread of a pollutant in an urban environment, first for inferring a two-dimensional parameter vector and second for a parameter field.

\tikzset{
  bigbox/.style={rectangle, draw, fill=gray!30, thick, inner sep=5pt, rounded corners},
  titlebox/.style={anchor=west}, 
  innerbox/.style={rectangle, draw, fill=white, thick, minimum width=1.5cm, minimum height=1.0cm,
                   rounded corners, align=left},
  widebox/.style={rectangle, draw, fill=white, thick, minimum width=8cm, minimum height=1.6cm,
                   rounded corners, align=left}
}

\begin{figure}
\footnotesize
\begin{tikzpicture}

\node[titlebox, text width = 16cm] (title) at (0,0) {\textbf{Optimal experimental design:}
Optimize controls to minimize posterior uncertainty such that path constraints are satisfied.};

\node[innerbox, below=0.25cm of title.south west, anchor=north west, align=left, text width = 4cm, draw=red] (in11) {
\textbf{Controls} $\alpha$ describe path $\ObsPath$ through an ODE
};

\node[innerbox, below=0.7cm of in11, text width = 4cm, draw=red] (in21) {
\textbf{Path constraints} have sensor avoid obstacles
};

\node[innerbox, below=0.7cm of in21, text width = 4cm] (in31) {
After optimization, \newline sensor takes measurements $\dobs$ along $\ObsPath$
};

\node[innerbox, right=0.5cm of in11.north east, anchor=north west, text width = 8cm, align = left, fill=gray!15] (bayes) {
\begin{minipage}{0.97\textwidth}
    \begin{tikzpicture}
    \node[titlebox, text width = \textwidth] (bayestitle) at (0,0) {\textbf{Bayesian inverse problem:} 
Given measurements $d_{\rm{obs}}$ along path $\mathbf{p}$, update the prior probability distribution of parameters $\mathbf{m}$.};
        \node[innerbox, text width = \textwidth, below = 0.1cm of bayestitle] (para2obs){
        \textbf{Parameter-to-observable map} $\mathcal{G}(\mathbf{m}; \mathbf{p})$ \\
        1. \textbf{Forward model} $\mathcal{F}$ solves a PDE \\
        2. \textbf{Simulate measurements} $d(t; \mathcal{F}(m), \mathbf{p})$ of the state $\mathcal{F}(\mathbf{m})$ along path $\mathbf{p}$
        };
        \node[innerbox, below=0.1cm of para2obs, text width = \textwidth] (noise){
        \textbf{Noise model} $d_{\rm{obs}} = d + \eta$, $\eta \sim \mathcal{N}(0, \mathcal{C}_{\rm{noise}})$ yields likelihood of simulated measurements $d$
        };
    \end{tikzpicture}
\end{minipage}
};

\node[innerbox, right=0.5cm of bayes.north east, anchor=north west, text width = 2.5cm, align = left] (posterior) {
Posterior\\ distribution
};

\node[innerbox, right=0.5cm of bayes.south east, anchor=south west, align=left, text width = 2.5cm, draw=red] (uq) {
\textbf{Posterior\\ uncertainty} is measured by \textbf{utility cost function} $\ObjOED$
};

\newdimen\xcoord
\newdimen\ycoord
\pgfextracty{\ycoord}{\pgfpointanchor{in11}{east}};
\pgfextractx{\xcoord}{\pgfpointanchor{bayes}{west}};
\coordinate (helper1) at ($(\xcoord, \ycoord)$);

\newdimen\xcoord
\newdimen\ycoord
\pgfextracty{\ycoord}{\pgfpointanchor{in31}{east}};
\pgfextractx{\xcoord}{\pgfpointanchor{bayes}{west}};
\coordinate (helper2) at ($(\xcoord, \ycoord)$);

\newdimen\xcoord
\newdimen\ycoord
\pgfextracty{\ycoord}{\pgfpointanchor{posterior}{west}};
\pgfextractx{\xcoord}{\pgfpointanchor{bayes}{east}};
\coordinate (helper3) at ($(\xcoord, \ycoord)$);

\draw[->, thick] (in11.south) -- (in21.north);
\draw[->, thick] (in21.south) -- (in31.north);
\draw[->, thick] (in11.east) -- (helper1);
\draw[->, thick] (in31.east) -- (helper2);
\draw[->, thick] (helper3) -- (posterior.west);
\draw[->, thick] (posterior.south) -- (uq.north);

\begin{pgfonlayer}{background}
  \node[bigbox, fit=(title) (in11) (in31) (bayes) (posterior) (uq), align = left] {};
\end{pgfonlayer}

\end{tikzpicture}

\caption{Overview of the components of the \ac{oed} problem and their connections.
The optimization problem is solved using an interior point method.
}
\label{fig:overview}
\end{figure}
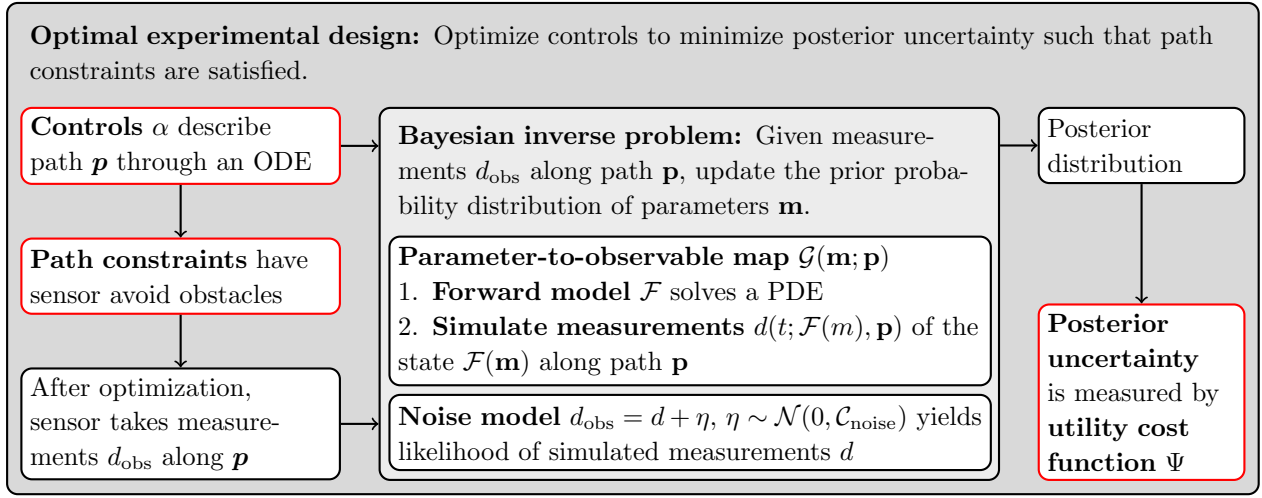

The remainder of this paper is organized as follows.
In Section \ref{sec:Fwd2OED}, we derive the Bayesian inverse problem for inferring parameters of interest from time-dependent measurement data measured by a sensor moving along an arbitrary path.
Next, in Section \ref{sec:oed}, we formulate our \ac{oed} objective function and associated minimization problem.
In Section \ref{sec:implementation}, we discretize the \ac{oed} problem, provide gradients, and explain our proposed optimization algorithm.
Finally, we demonstrate our proposed methodology in Section \ref{sec:numerics} on a convection-diffusion problem before concluding in Section \ref{sec:conclusion}.

\section{Bayesian inverse problem}\label{sec:Fwd2OED} 

In this section, we establish the mathematical setting and notation surrounding our \ac{pde}-constrained Bayesian inverse problem with measurements from a mobile sensor.
Because we consider linear \acp{pde} and we assume an additive Gaussian noise model below, the \ac{oed} formulations for the Bayesian and the deterministic inverse problem settings are structural analogues, and our \ac{oed} methodology can be applied to both with only slight differences in the interpretation of the involved matrices.

Let $\Domain \subset \R^{n_{\Domain}}$ with $n_{\Domain} \in \{2, 3\}$ be a bounded domain representing a physical area.
We are interested in \textit{states} $\u : [0, \finalt] \times \Domain \rightarrow \R$ that arise as solutions of a linear \ac{pde} for different parameters $\m \in \parameterspace$, with $\dimParameter \in \mathbb{N}$ and $\finalt > 0$ fixed.
To formalize this notion, let $\mathcal{V} \subset L^2(\Omega)$ be a Banach space on $\Domain$, and $L^2([0, \finalt], \mathcal{V})$ be the Bochner space of all functions $\u : [0, \finalt] \rightarrow \mathcal{V}$ such that $\|u\|_{L^2([0, \finalt], \mathcal{V})}^2 := \int_0^{\finalt} \|u(t)\|_{\mathcal{V}}^2 dt < \infty$ is finite.
A common choice would be, for instance, $\mathcal{V} = H^1(\Omega)$ the Sobolev space of one-time weakly differentiable functions.
We denote by $\pdesolutionoperator : \parameterspace \rightarrow \U \subset L^2([0, \finalt], \mathcal{V})$ the linear \ac{pde} solution operator that maps any parameter $\m \in \parameterspace$ onto the solution $u=\pdesolutionoperator(\m)$ of the \ac{pde}'s residual equation in $\mathcal{V}'$.
We assume that the domain $\Domain$ and the space $\U$ are chosen appropriately such that $\pdesolutionoperator$ is well-defined and continuous.

\begin{remark}
    We note that this setting permits stationary \acp{pde} on $\mathcal{V}$ by setting $\mathcal{U} := \{\u : [0, T] \rightarrow \mathcal{V} \text{ stationary }\}$.
\end{remark}

In this work, we consider the \ac{oed} of a mobile sensor whose path $\ObsPath : [0, \finalt] \rightarrow \R^{n_{\Domain}}$ is to be optimized.
To ensure that the sensor's path is physically feasible and remains within $\Domain$, we specifically consider paths from the set
\begin{equation}\label{eq:PathSpace}
    \begin{aligned}
    \PathSpace  := \{ &\ObsPath \in \mathcal{C}^1([0, \finalt]; \Omega) : ~ \exists \alpha \in \admissible \text{ s.t. } \ObsPath(0)=\ObsPath_0(\alpha), ~ \dot\ObsPath (t) = f(\alpha,\ObsPath(t),t) ~\forall t \in [0, T]  \}.
\end{aligned}
\end{equation}
Here, $\dot\ObsPath = f(\alpha,\ObsPath,t)$ with initial condition $\ObsPath(0)=\ObsPath_0(\alpha)$ is an \ac{ode} describing the dynamics of the physical sensor under varying control inputs $\alpha$.
The admissible set $\admissible$ for the controls $\alpha$ is chosen such that any $\ObsPath \in \PathSpace$ is considered realistic.
While the characterization of $\PathSpace$ can, in general, take other forms, expressing it through \acp{ode} is beneficial for optimizing over $\ObsPath \in \PathSpace$ numerically.
We discuss this claim in Section \ref{sec:optimization} below.

\begin{example}\label{ex:ODE}
    A common characterization of flight paths in the optimal control literature is through the velocity $v(t)$, angular velocity $\omega(t)$, and heading $\theta(t)$ of the mobile sensor, known as a \emph{unicycle kinematic model} (see \cite{unicyclerobot}, Section 13.3.1.1).
    For $\Omega \subset \R^2$, the relationship is expressed in the \ac{ode} system
    \begin{align}\label{ex:ode}
        \left(
        \begin{array}{c}
            \dot{\ObsPath}_1(t) \\
            \dot{\ObsPath}_2(t) \\
            \dot{\theta}(t)
        \end{array}
        \right) = 
        \left(\begin{array}{c}
            v(t) \cos(\theta(t)) \\
            v(t) \sin(\theta(t)) \\
            \omega(t)
        \end{array} \right)
    \end{align}
    with initial conditions $\ObsPath(0) = \mathbf{x}_0 \in \Omega$ and $\theta(0) = \theta_0 \in [0,2\pi)$.
    In our notation above, the control parameter becomes $\alpha = (\mathbf{x}_0, \theta_0, v, \omega)$ with $\ObsPath_0(\alpha):=\mathbf{\x}_0$.
    The admissible set $\admissible \subset \Omega \times \R \times L^2([0, \finalt]) \times L^2([0, \finalt])$ might then restrict the initial position and the minimum and maximum velocity and angular velocities according to the physical specifications of the sensing platform.
\end{example}

For an admissible path $\ObsPath \in \PathSpace$, we consider measurements of a state $\u \in \mathcal{U}$ at time $t \in [0, \finalt]$ to be of the (general) form
\begin{equation}\label{eq:data:general}
    \d(t; \u, \ObsPath) = \int_\Domain \u(t, \mathbf{x}) \MeasKern(\mathbf{x}, \ObsPath(t)) d\mathbf{x} 
\end{equation}
with a kernel function $\MeasKern : \Domain \times \Domain \rightarrow \R$.
The kernel function $\MeasKern$ defines a local weight of the values of $u(t, \,\cdot\,)$ around the position $\ObsPath(t)$ of the sensor at any time $t$, reflecting application-specific properties of the sensor such as detection radius or sensitivity.

\begin{example}\label{ex:gaussian} In general, it is unlikely that the sensor can follow the planned flight path exactly, e.g., due to navigation difficulties or imprecisions. We can include the uncertainty in its position in the measurement, for example, as
\begin{align}\label{eq:data:gaussian}
    \d(t; \u, \ObsPath) = (2\pi\sigma^2)^{-
    \dimDomain/2} \int_{\Domain} \exp\left( -\frac{1}{2\sigma^2} \|\mathbf{x}-\ObsPath(t)\|_2^2 \right) \u(t, \mathbf{x}) d\mathbf{x}
\end{align}
for a scaling variable $\sigma > 0$.
At any time $t$, the measurement data $\d(t; \u, \ObsPath)$ is the expectation $\mathbb{E}_{\mu}[\bar{\u}(t, \cdot)]$ of the extension $\bar{\u} : [0, \finalt]\times\mathbb{R}^{\dimDomain} \rightarrow \mathbb{R}$, $\left. \bar{u} \right|_{\Omega} = \u(t, \cdot)$, $\left. \bar{u} \right|_{\mathbb{R}^{\dimDomain} \setminus \Omega} = 0$ under the multivariate Gaussian $\mu = \mu(t, \ObsPath) = \mathcal{N}(\ObsPath(t), \sigma^2 \mathbf{I}_{\dimDomain})$ centered at $\ObsPath(t)$.
\end{example}

\begin{example}\label{ex:uniform}
To model the case where the position of the sensor is exact but its measurement is a spatial average over a surrounding ball $B_r(\ObsPath(t))$ with radius $r$,
we choose the measurement kernel $\MeasKern(\mathbf{x}, \mathbf{y}) = |B_{r}(0)|^{-1}\mathbbm{1}_{B_{r}(\mathbf{y})}(\mathbf{x})$.
Assuming the sensor's path $\ObsPath(t) \in \Omega$ remains far enough from the domain boundary $\partial \Domain$ such that 
$\text{dist}(\ObsPath(t), \partial \Omega) > r$, the measurement data
\begin{align}\label{eq:data:uniform}
    \d(t; \u, \ObsPath) = \frac{1}{|B_r(\ObsPath(t))|}\int_{B_r(\ObsPath(t))} \u(t, \mathbf{x}) d\mathbf{x} 
    = \mathbb{E}_{\mu}[u(t, \cdot)]
\end{align}
is the expectation of $u(t, \cdot)$ under the continuous uniform distribution $\mu = \mu(t, \ObsPath) = \mathcal{U}(B_r(\ObsPath(t)))$ on the ball $B_r(\ObsPath(t))$ centered around the position, $\ObsPath(t)$ of the moving sensor at time $t$.
\end{example}

We make the mild assumption that
\begin{align}\label{KernelAssumptions}
    0 < \esssup_{\mathbf{y} \in \Omega} \|\MeasKern(\,\cdot\,, \mathbf{y})\|_{L^2(\Omega)} < \infty,
\end{align}
which, firstly, ensures that definition \eqref{eq:data:general} defines $d(\,\cdot\,; \u, \ObsPath)$ as an element of $L^2([0, \finalt])$ for all $\u \in \mathcal{U}$ (c.f., Corollary \ref{thm:dinL2} in \ref{sec:proofs}); and secondly avoids a potential pitfall that arises if the measurement was defined over a set of measure zero, which requires more regularity of the state $\u$ in order for $\d$ to be meaningful.
Concatenated with the \ac{pde} solution operator $\pdesolutionoperator$ mapping parameters $\m$ onto states $\u = \pdesolutionoperator(\m) \in \U$, we obtain the \textit{parameter-to-observable} map $\paratoobs(\,\cdot\,;\ObsPath) : \parameterspace \rightarrow L^2([0, \finalt])$ defined as
\begin{align}\label{eq:paratoobs}
    \paratoobs(\m; \ObsPath) := \d(\, \cdot \, ;\pdesolutionoperator(\m), \ObsPath) \in L^2([0, \finalt])
\end{align}
for all $\m \in \parameterspace$.
Note that for any fixed $\ObsPath \in \PathSpace$, the parameter-to-observable map $\paratoobs(\,\cdot\,; \ObsPath)$ is a linear and continuous operator from the parameter space $\parameterspace$ into $L^2([0, \finalt])$.
Letting $\{\mbasis[i]\}_{i=1}^{\dimParameter}$ be the unit basis of $\mathbb{R}^{\dimParameter}$, we can thus write
\begin{align}\label{eq:paratoobs:linear}
    \paratoobs(\m; \ObsPath) = \sum_{i=1}^{\dimParameter} \m_i \d(\, \cdot \,; \pdesolutionoperator(\mbasis[i]), \ObsPath) = \sum_{i=1}^{\dimParameter} \m_i g_i(\ObsPath) \in L^2([0, \finalt]).
\end{align}
The abbreviation $g_i(\ObsPath) \in L^2([0, \finalt])$, defined as
\begin{align}\label{eq:unitoverservations}
    g_i(\ObsPath)(t) := \d(t; \pdesolutionoperator(\mbasis[i]), \ObsPath) = \int_\Domain \pdesolutionoperator(\mbasis[i])(t, \mathbf{x}) \MeasKern(\mathbf{x}, \ObsPath(t)) d\mathbf{x},
\end{align}
is the indirect observation of the $i$-th unit basis vector $\mbasis[i]$ under the \ac{pde}.

Following the Bayesian approach to inverse problems, we model the parameter of interest $\m$ as a random variable.
Moreover, we impose a Gaussian prior distribution $\m \sim \mathcal{N}(\mpr, \CovPr)$ with given prior mean $\mpr \in \parameterspace$ and prior covariance matrix $\CovPr \in \mathbb{R}^{\dimParameter \times \dimParameter}$.
We assume without loss of generality\footnote{
The case of a rank-deficient covariance matrix $\CovPr$ is recovered through restriction onto the orthogonal complement of its nullspace.
} that $\mathcal{N}(\mpr, \CovPr)$ is non-degenerate, i.e., $\CovPr \in \mathbb{R}^{\dimParameter \times \dimParameter}$ is invertible.
The overall goal of the Bayesian inverse problem is to decrease the uncertainty in the prior by informing it with observational data $\dobs$.
We model the observations received from the moving sensor along path $\ObsPath \in \PathSpace$ through an additive Gaussian noise model of the form
\begin{align}\label{eq:additiveNoise}
    \dobs = \dobs(\ObsPath) = \paratoobs(\m; \ObsPath) + \eta
\end{align}
where $\eta \in L^2([0, \finalt])$, $\eta \sim \mathcal{N}(0, \OpCovNoise)$ with mean $0 \in \Ltwotime$ and a covariance operator $\OpCovNoise : L^2([0, \finalt]) \rightarrow L^2([0, \finalt])$, i.e., $\OpCovNoise$ is a self-adjoint, symmetric positive semi-definite trace-class operator (\cite{stuart2010inverse}, Theorem 6.6).
For simpler exposition below, we assume in the following without loss of generality\footnote{The arguments still apply to the general setting by restriction onto the Cameron-Martin space $\text{Im}(\OpCovNoise^{1/2})$, see \cite{stuart2010inverse}, p. 530.} that 
$\OpCovNoise$ is indeed symmetric positive-definite.

\begin{remark}
    Considering the observational data as a function in $L^2([0,T])$ is sensible when the frequency at which data is measured and saved is much finer than the resolution needed to resolve the path; see \cite{ha2019periodic} (Section 1) for a discussion.
    The ``discrete time'' setting of measurements at distinct points $0<\tau_1< \dots < \tau_{n_{\tau}} \le T$ for some $n_{\tau} \in \mathbb{N}$ can be treated by restricting \eqref{eq:data:general} to the time steps $t = \tau_1, \dots, \tau_{n_{\tau}}$, changing the space of measurement space from $L^2([0, \finalt])$ to $\R^{n_{\tau}}$, and changing the noise model to $\mathcal{N}(\mathbf{0}, \CovNoise)$ with symmetric positive definite noise covariance matrix $\CovNoise \in \R^{n_{\tau} \times n_{\tau}}$.
\end{remark}

With the additive Gaussian noise in \eqref{eq:additiveNoise},
the likelihood of a parameter $\m$ given noisy observational data $\dobs$ is proportional to
\begin{equation}\label{eq:likelihood}
    \likelihood(\m) \propto \exp\left( -\frac12 \| \OpCovNoise^{-1/2}(\paratoobs(\m; \ObsPath) - \dobs) \|_{\Ltwotime}^2 \right). 
\end{equation}
Consequently, following Bayes Theorem (e.g., \cite{stuart2010inverse}, Section 2.2), the posterior measure $\postMeasure$ is proportional to
\begin{align}\label{eq:postMeasure}
    \postMeasure \propto \exp \left( -\| \OpCovNoise^{-1/2}(\paratoobs(\m; \ObsPath) - \dobs) \|_{\Ltwotime}^2 - \|\m - \mz\|^2_{\CovPr^{-1}} \right)
\end{align}
More specifically, using that the parameter-to-observable map $\paratoobs(\ObsPath)$ is linear, $\postMeasure$ is a Gaussian $\mathcal{N}(\mpost, \CovPost)$ with
posterior mean and posterior covariance, 
\begin{align}
    \mpost &:= \CovPost (\paratoobs(\ObsPath)^{*} \OpCovNoise^{-1} \dobs + \iCovPr \mz) \in \parameterspace\\
    \CovPost &:= \left(\paratoobs(\ObsPath)^{*} \OpCovNoise^{-1} \paratoobs(\ObsPath) + \iCovPr\right)^{-1} \in \mathbb{R}^{\dimParameter \times \dimParameter}, \label{eq:CovPost}
\end{align}
where $\paratoobs(\ObsPath)^{*} : L^2([0, \finalt]) \rightarrow \parameterspace$ denotes the adjoint operator operator of $\paratoobs(\ObsPath)$.
The first term in \eqref{eq:CovPost},
\begin{align}\label{eq:Fisher}
    \Fisher := \paratoobs(\ObsPath)^{*} \OpCovNoise^{-1} \paratoobs(\ObsPath) \in \mathbb{R}^{\dimParameter \times \dimParameter},
\end{align}
is the Fisher information matrix, i.e., the second derivative of the negative log-likelihood function $-\log(\likelihood)$.
Inserting the representation \eqref{eq:paratoobs:linear} of $\paratoobs(\ObsPath)$, we get that
\begin{align}\label{eq:Fisher:entries}
    \left[\Fisher\right]_{i,j} 
    = \left<\OpCovNoise^{-1/2} g_i(\ObsPath), \OpCovNoise^{-1/2} g_j(\ObsPath)\right>_{L^2([0, \finalt])}
\end{align}
for $1 \le i, j \le \dimParameter$.
The fact that the forward model is linear implies that the posterior covariance matrix does not depend on observational data, or on the posterior mean. 
As such, we can solve the \ac{oed} problem efficiently, which would otherwise be a bilevel optimization problem.

\section{Optimal Experimental Design}
\label{sec:oed}

The uncertainty in the posterior distribution $\mathcal{N}(\mpost, \CovPost)$ is characterized by the posterior covariance matrix $\CovPost$, and, in particular, its eigenvalues $\lambda_1(\ObsPath) \ge \dots \ge \lambda_{\dimParameter}(\ObsPath) \ge 0$.
Eigenvectors associated to large eigenvalues point into the directions in which the posterior mean is particularly uncertain; these are also the directions along which changes to the parameter have little effect on the observational data (relative to their prior uncertainty).
In contrast, eigenvectors associated to small eigenvalues point into the directions of small uncertainty, either because their prior uncertainty was already small or because the parameter-to-observable map is sensitive into these directions, yielding relatively large changes in observational data already for little parametric changes.
To quantify the overall uncertainty in the posterior, 
the literature distinguishes various \textit{utility cost functions} $\ObjOED : \mathbb{R}^{\dimParameter \times \dimParameter} \rightarrow \mathbb{R}$ (see \cite{ucinski2004optimal} for an overview).
The \textit{optimal experimental design} is then identified as the minimizer of the chosen utility cost function.
Here, 
we  consider the following \ac{oed} problem to find the optimal path $\ObsPath \in \mathcal{C}([0, \finalt], \Domain)$ 
\begin{align}\label{eq:OED:minimization:infeasible}
    \min_{\ObsPath \in \PathSpace} \ObjOED(\CovPost).
\end{align}
We restrict our exposition to two common choices for $\ObjOED$: the A- and D-\ac{oed} utility functions
\begin{align}
    \ObjAOED(\CovPost) &:= \tr\left( \CovPost \right) = \sum_{i=1}^\dimParameter \lambda_i(\ObsPath), &
    \ObjDOED(\CovPost) &:= \det\left( \CovPost \right) = \prod_{i=1}^\dimParameter \lambda_i(\ObsPath).
\end{align}
Interpreted in terms of uncertainty, the A-\ac{oed} utility $\ObjAOED(\CovPost)$ penalizes the mean variance of posterior samples, while the D-\ac{oed} utility $\ObjDOED(\CovPost)$ penalizes the volume of the uncertainty ellipsoid described by $\CovPost$ (see \cite{ucinski2004optimal}, chapter 2.3).
In our notation hereafter, we use $\ObjAOED$ and $\ObjDOED$ only when specifically distinguishing between the \ac{oed} criteria, and otherwise use $\ObjOED$ to denote that a statement holds in generality.

In addition to ensuring that the path's dynamics remain physically consistent through $\ObsPath \in \PathSpace$, one needs to ensure that it also avoids obstacles. 
We assume for the sake of simplicity that any obstacle 
can be enclosed in a ball or cube of suitable size (or a combination thereof), so that we can write obstacle constraints as
\begin{align}
    \|\ObsPath(t) - \bm{c}\|_\mathbf{W}^2 & \ge 1, & \text{(ellipsoid)}\\
    \|\mathbf{T} \ObsPath(t) - \bm{c}\|_\infty & \ge 1, &\text{(rectangle)}
\end{align}
where $\bm{c} \in \R^{\dimDomain}$ is the constraint's center, $\mathbf{W} \in \mathbb{R}^{\dimDomain \times \dimDomain}$ is symmetric positive definite, and $\mathbf{T} \in \mathbb{R}^{\dimDomain \times \dimDomain}$ is a linear transformation matrix (scaling, rotation, reflection, shear, etc.) that maps the rectangle that we wish to exclude to the unit rectangle.
This representation lets us summarize all positional restrictions ---including $\ObsPath(t) \in \Domain$ for all $t \in [0, \finalt]$--- in the \textit{path constraints}
\begin{equation}\label{eq:pathCons}
    c(\ObsPath,t) \geq 0 \quad \forall t \in [0, \finalt],
\end{equation}
for some function $c : \mathcal{C}([0, \finalt], \R^{\dimDomain}) \to \R^{n_{\rm{constr}}}$ with $n_{\rm{constr}} \in \mathbb{N}$ being the number of constraints.
We note that the cube constraint is not differentiable at its center $\bm{c}$ due to the use of the $\ell_{\infty}$ norm. However, it is always continuously differentiable in a neighborhood of the feasible set, because the right-hand-side is strictly positive.

In addition, we add a regularization term, $\regularization(\alpha, \ObsPath)$, to the cost function that penalizes
positional (e.g., closeness to obstacles) and control-related (e.g., sudden acceleration) considerations separately. Thus,
we arrive at the constrained, regularized minimization problem
\begin{subequations}\label{eq:flightpathOED:abstract}
\begin{align}
    &\mini_{\alpha \in \admissible, \ObsPath \in \mathcal{C}^1([0, \finalt], \R^{\dimDomain})} 
        && \ObjOED(\CovPost) + \regularization(\alpha, \ObsPath) \\
    &\text{subject to} 
        &&\dot\ObsPath = f(\alpha,\ObsPath,t), \quad \text{a.e. in } [0, \finalt] && \text{(control parameterization, \eqref{eq:PathSpace})}\\
        & && \ObsPath(0)=\ObsPath_0(\alpha), &&\text{(initial position, \eqref{eq:PathSpace})}\\
     &   && 0 \leq c(\ObsPath,t), \quad \text{a.e. in } [0, \finalt] && \text{(path constraints, \eqref{eq:pathCons})} 
\end{align}    
\end{subequations}
where the posterior covariance matrix, $\CovPost$, is defined via the following set of equations almost everywhere (a.e.) $t \in [0, \finalt]$ and for all $1 \le i \le \dimParameter$:
\[
    \begin{array}{ll}
        \dps \CovPost = \big[\Fisher + \iCovPr\big]^{-1} & \text{(post. cov. matrix, \eqref{eq:CovPost} and \eqref{eq:Fisher})} \\
        \dps \left[\Fisher\right]_{i,j} = \int_{[0, \finalt]} (\OpCovNoise^{-1/2} g_i(\ObsPath)(s)) (\OpCovNoise^{-1/2}g_j(\ObsPath)(s)) ds & \text{(Fisher information matrix, \eqref{eq:Fisher:entries})} \\
        \dps g_i(\ObsPath)(t) = \int_{\Omega} \pdesolutionoperator(\mbasis[i]) \MeasKern(\mathbf{y}, \ObsPath(t)) d\mathbf{y} & \text{(unit observations under \ac{pde}, \eqref{eq:unitoverservations})} 
    \end{array}
\]
where $\{\mbasis[i]\}_{i=1}^{\dimParameter}$ is the unit basis of $\R^{\dimParameter}$.

We note that \eqref{eq:flightpathOED:abstract} differs from a classical optimization problem in the sense that the definition of the objective function requires the assembly of the covariance matrix, $\CovPost$, and the computation of the eigenvalues of its inverse, which cannot in general be expressed as explicit algebraic equations. We show below that we can nevertheless apply classical optimization techniques because we can compute gradients of $\CovPost$ with respect to the path $\ObsPath$ using the chain rule.

\section{Computational approach}\label{sec:implementation}

We next describe how to solve the \ac{oed} minimization problem \eqref{eq:flightpathOED:abstract} numerically.
We proceed in the following steps:
First, in Section \ref{sec:discretization}, we discretize the posterior covariance matrix $\CovPost$.
In Section \ref{sec:optimization}, we state the discrete optimality system and discuss our proposed algorithm for solving it.
Finally, in Section \ref{sec:gradients:discrete}, we determine the gradients of the A- and D-\ac{oed} utility values $\ObjAOED(\CovPost)$ and $\ObjDOED(\CovPost)$ with respect to the path $\ObsPath$ as required by the optimization algorithm.
Where possible, we denote discretized variables in boldface, with vectors in lower-case and matrices in upper-case; where this is not possible, we use the superscript ``\textit{h}'' to denote discrete.

\subsection{Discretization}\label{sec:discretization}

To approximate the posterior $\postMeasure = \mathcal{N}(\mpost, \CovPost)$ numerically, we first discretize the forward model to obtain a discrete solution operator $\pdesolutionoperatorDiscrete : \parameterspace \rightarrow \UDiscrete$ mapping any $\m \in \parameterspace$ onto an approximation $\pdesolutionoperatorDiscrete(\m) \approx \pdesolutionoperator(\m)$ within a finite-dimensional subspace $\UDiscrete \subset \U \subset L^2([0, \finalt], \mathcal{V})$.
We require that the discretized \ac{pde} solution $\pdesolutionoperatorDiscrete(\m)$ can be evaluated at any time $t\in[0, \finalt]$, which can be achieved, for instance, by interpolation between steps of a time stepping scheme, or through the use of a space-time finite element discretization.
Note that we do not prescribe a specific time discretization for the forward model, as we consider it an external input that the \ac{oed} problem builds upon.

Next, we introduce a discretization of the time domain $[0, \finalt]$ into $\dimTime \in \mathbb{N}$ time steps $0 = \tk[1] < \tk[2] < \dots < \tk[\dimTime] = \finalt$.
For ease of exposition, we let the time steps be equidistant, i.e., $\tk = (k-1) \dt$ with $\dt := T/(\dimTime-1)$.
We solve the \ac{ode} numerically through a time-stepping scheme, denoting the obtained time steps by $\pk^h \approx \pk := \ObsPath(\tk)$, $1\le k \le \dimTime$.
We define $\ObsPath^h := \{\pk^h\}_{k=1}^{\dimTime}$.

As the first step in the discrete approximation to the parameter-to-observable map $\paratoobs$, we define the matrix $\paratoobsDiscrete = [\mathbf{g}_1, \dots, \mathbf{g}_{\dimParameter}] \in \mathbb{R}^{\dimTime \times \dimParameter}$,
\begin{align}\label{eq:paratoobs:discrete}
    [\paratoobsDiscrete]_{k,m} := [\mathbf{g}_m]_k := \int_{\Domain} \pdesolutionoperatorDiscrete(\mbasis)(\tk, \mathbf{x}) \MeasKern(\mathbf{x}, \pk^h)d\mathbf{x}
\end{align}
where $\mbasis \in \mathbb{R}^{\dimParameter}$, $(\mbasis)_n = \delta_{n,m}$ are the unit basis vectors of $\parameterspace$.
The columns $\mathbf{g}_m \in \mathbb{R}^{\dimTime}$ of $\paratoobsDiscrete$ are ordered vectors containing the measurements of $\pdesolutionoperator(\mbasis)$ at the positions $\pk^h$, $1\le k \le \dimTime$.
To map the measurements at these discrete locations back into $L^2([0, \finalt])$, we interpolate linearly between them.
To this end, we let $\febasis : [0, \finalt] \rightarrow \mathbb{R}$ be the unique piecewise linear finite element basis functions with $\febasis(\tk[\ell]) = \delta_{k, \ell}$, and define
\begin{align}\label{eq:discrete:dh}
    g^h_m(t; \ObsPath^h) := \sum_{k=1}^{\dimTime} \febasis(t) [\mathbf{g}_m]_k \approx g_m(t; \ObsPath^h)
\end{align}
as an approximation to the indirect observations $g_m$ of the unit basis from \eqref{eq:unitoverservations}.
This gives us the discrete parameter-to-observable map $\paratoobsDiscreteh(\ObsPath^h) : \parameterspace \rightarrow L^2{[0,\finalt]}$ through
\begin{align}
    \paratoobsDiscreteh(\m; \ObsPath^h)(t) 
    := \sum_{m=1}^{\dimParameter} \m_m g^h_m(t; \ObsPath^h)
    = \sum_{m=1}^{\dimParameter} \m_m \sum_{k=1}^{\dimTime} \febasis(t) [\mathbf{g}_m]_k
    = \sum_{k=1}^{\dimTime} \febasis(t) [\paratoobsDiscrete \m]_{k},
\end{align}
in a decomposition analogous to \eqref{eq:paratoobs:linear}.

\begin{proposition}\label{thm:continuity}
    Suppose $\pk^h \in \Domain$ for all $k \in \{1, \dots, \dimTime\}$ and $\pdesolutionoperator : \parameterspace \rightarrow \mathcal{C}([0, \finalt], L^2(\Domain))$.
    Suppose further that $\MeasKern(\mathbf{x}, \, \cdot \,) \in L^2(\Domain)$ is Lipschitz-continuous with Lipschitz continuity constant $\gamma_{\MeasKern}(\mathbf{x}) \le \overline{\gamma}$ for almost all $x \in \Domain$ and an upper bound $0 < \overline{\gamma} < \infty$.
    Then the matrix $\paratoobsDiscrete$ is continuous in each $\pk$, $k\in \{1, \dots, \dimTime\}$.
\end{proposition}
\begin{proof}
    The proof is provided in \ref{sec:proofs}.
\end{proof}

To approximate the Fisher information matrix $\Fisher$, we insert the representation \eqref{eq:discrete:dh} of $g^h_m$ in terms of the temporal basis functions $\psi_k$ into \eqref{eq:Fisher:entries}.
We obtain
\begin{equation}\label{eq:Fisher:discrete}
  \begin{aligned}
    \left[\FisherDiscrete\right]_{i,j} &= \int_{[0, \finalt]} \left(\OpCovNoise^{-1/2} g_i^h(t; \ObsPath^h)\right) \left(\OpCovNoise^{-1/2}g_j^h(t; \ObsPath)\right) dt \\
    &= \sum_{k,\ell=1}^{\dimTime} g_i^h(\tk; \ObsPath^h) g_{j}^h(\tk[\ell]; \ObsPath^h) \int_{[0, \finalt]} \left(\OpCovNoise^{-1/2} \psi_k(t)\right) \left(\OpCovNoise^{-1/2}\psi_{\ell}(t)\right) dt \\
    &= \mathbf{g}_i^T \CovNoise^{-1} \mathbf{g}_j \\
    &= [\paratoobsDiscrete^{\top}\CovNoise^{-1} \paratoobsDiscrete]_{i,j}
\end{aligned}  
\end{equation}
with (inverse) noise covariance matrix $\CovNoise^{-1} \in \mathbb{R}^{\dimTime \times \dimTime}$ defined as
\begin{align}\label{eq:CovNoise:inverse}
    \left[\CovNoise^{-1}\right]_{i,j} := \int_{[0, \finalt]} \left(\OpCovNoise^{-1/2} \psi_k(t)\right) \left(\OpCovNoise^{-1/2}\psi_{\ell}(t)\right) dt.
\end{align}
Note that $\CovNoise^{-1}$ as defined here is indeed invertible because the covariance operator $\OpCovNoise$ is symmetric positive-definite.
Definition \eqref{eq:CovNoise:inverse} requires knowledge of the \textit{inverse} noise covariance \textit{operator} $\OpCovNoise^{-1}$; this can be circumvented by following a discretize-then-optimize approach where $\OpCovNoise$ is discretized first to a matrix $\CovNoise$.
However, if $\OpCovNoise$ is defined as the weak solution operator to a \ac{pde} then $\OpCovNoise$ can simplify such that $\CovNoise^{-1}$ has indeed a closed-form expression that can be computed without taking the matrix-inverse.
We illustrate this claim in the following example.

\begin{example}\label{ex:noise:2}
    In our numerical experiments, we choose $\OpCovNoise$ as the  operator mapping any $f \in \Ltwotime$ onto the weak solution $v \in H^1([0, \finalt])$ of the \ac{pde}
    \begin{align}\label{ex:noise:pde:strong}
        -\Delta v + 100 v = f, \quad v'(0)=v'(\finalt)=0
    \end{align}
    Following \cite{stuart2010inverse} (Example 6.17) and \cite{BuiThanh2013}, this defines a Gaussian in $\Ltwotime$. 
    Consequently, for $v, \psi \in H^1([0,\finalt])$ holds
    \begin{align*}
        \left(\OpCovNoise^{-1} v, \psi \right)_{L^2([0,\finalt])}
        &= \int_{[0,\finalt]} \nabla v \cdot \nabla \psi + 100 v \psi dt.
    \end{align*}
    By using that $\OpCovNoise^{-1/2}$ is self-adjoint and that $\psi_i \in H^1([0, \finalt])$ holds for the temporal basis functions in our discretization, we get a closed-form expression for $\CovNoise^{-1}$.
\end{example}

With the discrete approximation \eqref{eq:Fisher:discrete} to the Fisher information matrix, we finally arrive at the discrete approximation to the posterior covariance matrix:
\begin{align}\label{eq:CovPost:discrete}
    \CovPostDiscrete :=
    \left(\FisherDiscrete + \iCovPr\right)^{-1}
    =\left(\paratoobsDiscrete^{\top}\CovNoise^{-1} \paratoobsDiscrete + \iCovPr\right)^{-1}.
\end{align}
By construction, $\FisherDiscrete + \iCovPr \in \mathbb{R}^{\dimParameter \times \dimParameter}$ is symmetric positive-definite, and the use of the matrix inverse in the definition \eqref{eq:CovPost:discrete} is thus justified.

\begin{corollary}
    Under the assumptions of Proposition \ref{thm:continuity}, $\CovPostDiscrete$ is continuous in each $\pk$, $k \in \{1, \dots, \dimTime\}$.
\end{corollary}

\begin{proof}
    The proof follows immediately from the continuity of $\paratoobsDiscrete$ in $\pk$ and because the eigenvalues of $\paratoobsDiscrete^{\top}\CovNoise^{-1} \paratoobsDiscrete + \iCovPr$ are bounded from below by those of $\iCovPr$.
\end{proof}

In general, we do not invert the matrix on the right-hand side, but either work with a factorization or with the reciprocal of the eigenvalues of $\CovPostDiscrete^{-1}$:
If $\CovNoise^{-1}$ has a sparse closed-form expression (as in Example \ref{ex:noise:2}), then the matrix action of $\FisherDiscrete$ can be computed in $\mathcal{O}(\dimTime \dimParameter)$ from $\paratoobsDiscrete \in \mathbb{R}^{\dimTime \times \dimParameter}$.
The eigenvalues of $\CovPost$ can then be computed iteratively without explicitly forming $\Fisher$, $\iCovPr$, or $\CovPostDiscrete^{-1}$.
Since the A- and D-\ac{oed} utility cost functions require all $\dimParameter$ eigenvalues, the cost of evaluating $\ObjAOED(\CovNoise)$ and $\ObjDOED(\CovNoise)$ then scales as $\mathcal{O}(\dimTime \dimParameter^2)$.
However, this order is contingent on the availability of $\paratoobsDiscrete$, which naturally incurs additional compute time, specifically from two sources: first, the evaluation of the \ac{pde} for all parameters $\m \in \parameterspace$, and second, the computation of the measurement data $\mathbf{d}^h_m(\ObsPath^h)$ at the $\dimTime$ discrete positions $\pk^h$, $1\le k \le \dimTime$.
Here, we note first that if the parameter space dimension $\dimParameter$ is reasonably small, the states $\{\pdesolutionoperatorDiscrete(\mbasis)\}_{m=1}^{\dimParameter}$ can be precomputed and stored;
the matrix $\paratoobsDiscrete$ can then be assembled for any $\ObsPath^h$ without requiring any additional \ac{pde} solves. 
For high-dimensional parameter spaces, we suggest to perform either a low-rank approximation of $\pdesolutionoperatorDiscrete$ or $\CovPostDiscrete$ (see \cite{flath2011fast}), or a parameter space reduction (see \cite{cui2016scalable} and Section \ref{sec:results:transient:highdim}) as a pre-processing step to restrict the required number of \ac{pde} solves.

With the \ac{pde} solves outsourced to a preparatory, path-independent step, the dominant computational cost for the evaluation of $\paratoobsDiscrete$ is the computation of the $\dimTime \dimParameter$ measurements $[\mathbf{d}_k]_m \in \mathbb{R}$, each requiring the evaluation of an integral \eqref{eq:paratoobs:discrete} over $\Domain$.
There are several approaches for dealing with this computational cost, amongst others:
To precompute the convolution $\u^h(\tk, \, \cdot \,) \ast \MeasKern( \, \cdot \,, \mathbf{y})$ as a function of $\mathbf{y} \in \Domain$ thereby reducing the evaluation of the integrals to point evaluations;
to save once computed measurements for future iterations of the optimization;
to use known measurements at close locations as approximates;
to limit the integration to the support of $\MeasKern$; 
or to approximate the integral by a pointwise evaluation.
For example, in our numerical demonstration in Section \ref{sec:numerics}, we first approximate measurements of the form \eqref{eq:data:uniform} through point-evaluations of $\u^h(t)$ such that $\paratoobsDiscrete$ can be computed fast for many paths $\ObsPath^h$.
After solving the optimization problem for pointwise measurements, we use the obtained path and its adjoint variables as educated initial guesses for solving the optimization problem again (warm-starting), this time evaluating the integrals \eqref{eq:data:uniform} on a fine mesh of the support of $\MeasKern(\, \cdot \,, \pk^h)$ (see Section \ref{sec:results:transient:warmstart}).

\subsection{Discretized optimization problem}\label{sec:optimization}

Next, we use the discretizations developed above to formulate a discretized version of the optimal experimental design problem \eqref{eq:flightpathOED:abstract} building on \eqref{eq:paratoobs:discrete}--\eqref{eq:CovPost:discrete}. 
Replacing the posterior covariance matrix $\CovPost$ for its discrete approximation $\CovPostDiscrete$ in \eqref{eq:flightpathOED:abstract} we obtain
\begin{subequations}\label{eq:FiniteDimOED}
    \begin{align}
        \mini_{\alpha, \ObsPath^h} \quad & \dps \ObjOED\big(\CovPostDiscrete\big) +
        R(\alpha, \ObsPath^h) \label{eq:FiniteDimOED:a}\\
        \text{subject to} \quad  & \DAEPath_k^h - \DAEPath_{k-1}^h  = \dt f(\DAEPath_{k-1}^h, \alpha, t_{k-1}), \quad \forall k=2,\ldots,n_t, \label{eq:DiscrODE}\\
        & \DAEPath_1^h = \mathbf{p}_0(\alpha) \\
        & 0 \le c(\DAEPath_k^h, t_k), \quad \forall k=1,\ldots,n_t, \\
        & \alpha \in \admissible, 
    \end{align}
\end{subequations}
where we have used explicit Euler time-stepping to illustrate a possible discretization of the \ac{ode} $\dot{\ObsPath} = f(\ObsPath, \alpha, t)$ in \eqref{eq:DiscrODE}.
Naturally, the evaluation of the constraint $\alpha \in \admissible$ and the computation of the regularization term $R(\alpha, \ObsPath)$ depend on the interpretation of the control parameter $\alpha$ within the target application scenario, and \eqref{eq:FiniteDimOED} is thus still general.

To make \eqref{eq:FiniteDimOED} more concrete, we expand upon the control problem \eqref{ex:ode} introduced in Example \ref{ex:ODE} with $\alpha = (\mathbf{x}_0, \theta_0, v, \omega) \in \admissible = \Omega \times \R \times L^2([0, \finalt]) \times L^2([0, \finalt])$ controlling the initial position $\mathbf{x}_0$, the initial heading $\theta_0$, the velocity $v$, and the angular velocity $\omega$ of the measurement platform.
We discretize $v$ and $\omega$ as piecewise constants with $v_k = v(\tk)$ and $\omega_k = \omega(\tk)$ using the same time discretization as for the path $\ObsPath^h$.
To penalize high speeds, sharp turns, and sudden or jerky maneuvers, we define the quadratic regularization function
\begin{align*}
    R(\alpha) &:= \|v^h\|_{L^2([0, \finalt])}^2 +  \|\omega^h\|_{L^2([0, \finalt])}^2 + \|\omega^h\|_{2-var}^2 \\
    &= \dt \sum_{k=1}^{\dimTime-1} v_k^2 + \dt \sum_{k=1}^{\dimTime-1} \omega_k^2 + \sum_{k=1}^{\dimTime-1} (\omega_{k+1}-\omega_k)^2.
\end{align*}
where $\| \, \cdot \, \|_{2-var}$ denotes the 2-variation norm (see \cite[Section 5.1]{friz2010multidimensional}) which measures the accumulated (squared) jump distance of discontinuities.
With this control parameterization and regularization, \eqref{eq:FiniteDimOED} takes the form
\begin{subequations}\label{eq:FiniteDimOED:concrete}
    \begin{align}
        \mini_{v^h, \omega^h, \ObsPath^h} \quad & \dps \ObjOED(\CovPostDiscrete) +
        \dt \sum_{k=1}^{\dimTime-1} v_k^2 + \dt \sum_{k=1}^{\dimTime-1} \omega_k^2 + \sum_{k=1}^{\dimTime-1} (\omega_{k+1}-\omega_k)^2 \label{eq:FiniteDimOED:concrete:a}\\
        \text{subject to} \quad  
        & \ObsPath_1 = \mathbf{x}_0 \label{eq:InitCond}\\
        & \ObsPath_{k} - \ObsPath_{k-1} = \dt \left( \begin{array}{c}
        v_{k-1} \cos(\theta_{k-1}) \\
        v_{k-1} \sin(\theta_{k-1})
        \end{array}\right) \quad \forall k=2,\ldots,n_t \\
        & \theta_1 = \theta_0 \\
        & \theta_k - \theta_{k-1} = \dt \, \omega_{k-1} \quad \forall k=2,\ldots,n_t \\
        & 0 \le c(\DAEPath_k^h, t_k), \quad \forall k=1,\ldots,n_t, \label{eq:PathCons}
    \end{align}
\end{subequations}
where $\theta_2, \dots, \theta_{\dimTime}$ are auxiliary variables, and we have imposed the initial position $\mathbf{x}_0$ and the initial heading $\theta_0$ as constraints without loss of generality.

The discretized optimization problem \eqref{eq:FiniteDimOED} contains two classes of terms. The first class are explicit algebraic constraints, \eqref{eq:InitCond}--\eqref{eq:PathCons}, that represent the discretization of the ODE and the regularization term, $\regularization(\alpha,\ObsPath)$. These terms are handled explicitly by nonlinear optimization solvers such as \ac{ipopt} \cite{wachter2006implementation}. The second class are implicitly-defined terms, namely the OED objective, $\ObjOED(\CovPostDiscrete)$. To evaluate the OED objective for a given control, $\alpha$ and path $\DAEPath$, we first need to assemble the Fisher information matrix, $\Fisher$, and then factorize the posterior, $\CovPostDiscrete$. We can then form the gradient of $\ObjOED(\CovPostDiscrete)$ with respect to $\DAEPath$, by differentiating through the linear system solve. We note, that the cost of evaluating the implicit terms is orders of magnitude larger than the cost of evaluating the explicit terms, and that the evaluation of Hessian or second-order derivatives of the implicit terms would be prohibitive. Hence, we use a quasi-Newton approximation of the Hessian of the Lagrangian, which allows us to solve the optimization problem with first-order information only. Because all functions (explicit and implicit) are smooth, we expect the nonlinear solvers to converge to a stationary point of \eqref{eq:FiniteDimOED}.

\subsection{Gradient computations}\label{sec:gradients:discrete}

The solution of the discretized optimization problem \eqref{eq:FiniteDimOED} requires the computation of the gradient of the \ac{oed} utility value $\ObjOED(\CovPostDiscrete)$ with respect to the positions $\pk^h$, $1\le k \le \dimTime$, that the mobile sensor takes for each time step.
The \ac{oed} objective depends on the eigenvalues of the posterior covariance matrix. In general, eigenvalues are not differentiable with respect to matrix coefficients (whenever the eigenvalue has multiplicity greater than one). However, the trace and determinant of a matrix are differentiable as long as the posterior covariance is invertible (our covariance matrices are positive definite and invertible). The gradients  can be evaluated using the chain rule, as shown in Proposition~\ref{thm:gradient:1}.

\begin{proposition}\label{thm:gradient:1}
    Suppose $\ObsPath^h \subset \Domain$ and $\pdesolutionoperator^h : \parameterspace \rightarrow \mathcal{C}([0, \finalt], L^2(\Domain))$.
    Suppose further that $\MeasKern(\mathbf{x}, \, \cdot \,) \in L^2(\Domain)$ is Lipschitz-continuous with Lipschitz continuity constant $\gamma_{\MeasKern}(\mathbf{x}) \le \overline{\gamma}$ for almost all $x \in \Domain$ and an upper bound $0 < \overline{\gamma} < \infty$.
    For any $k \in \{1, \dots, \dimTime\}$ holds:
    If the parameter-to-observable matrix $\paratoobsDiscrete$ is continuously differentiable with respect to $\pk^h$, then the A- and D-\ac{oed} utility costs $\ObjAOED(\CovPostDiscrete)$ and $\ObjDOED(\CovPostDiscrete)$ are continuously differentiable with respect to $\pk^h$.
    In this case, for direction $i\in \{1, \dots, \dimDomain\}$, 
    \begin{align}
        \left[ \nabla_{\pk} \ObjAOED(\CovPostDiscrete) \right]_i &= -2 \operatorname{trace} \left( \CovPostDiscrete  \paratoobsDiscrete^{\top}\CovNoise^{-1}\frac{d \paratoobsDiscrete}{d \pki} \CovPostDiscrete \right)\label{eq:derivative:A:discrete}
    \end{align}
    for A-\ac{oed} and
    \begin{align}
    \left[\nabla_{\pk} \ObjDOED(\CovPostDiscrete)\right]_i
    &= -2\ObjDOED(\CovPostDiscrete) \operatorname{trace}\left( \paratoobsDiscrete^{\top}\CovNoise^{-1}\frac{d \paratoobsDiscrete}{d \pki} \CovPostDiscrete \right)\label{eq:derivative:D:discrete}
\end{align}
for D-\ac{oed}.
\end{proposition}

\begin{proof}
The proof is provided in \ref{sec:proofs}.
\end{proof}

The critical assumption in Proposition \ref{thm:gradient:1} is the differentiability of $\paratoobsDiscrete$, which depends immediately on the chosen measurement kernel function $\MeasKern$.
The following corollary provides a sufficient criterion.

\begin{corollary}\label{thm:gradient:2}
Let the assumptions of Proposition \ref{thm:gradient:1} hold.
Assume additionally that the measurement kernel function $\MeasKern : \Domain \times \Domain \rightarrow \mathbb{R}$ is continuously differentiable with respect to its second argument with 
\begin{align}
    \esssup_{\mathbf{y} \in \Omega} \| \nabla_{\mathbf{y}} \MeasKern(\cdot, \mathbf{y})\|_{L^2(\Omega)} < \infty.
\end{align}%
Then the parameter-to-observable matrix $\paratoobsDiscrete$ is continuously differentiable with respect to $\pk$ for all $k \in \{1, \dots, \dimTime\}$ with
\begin{align}\label{eq:deriv:G:1}
    \left[\frac{d \paratoobsDiscrete}{d \pki} \right]_{\ell, m} 
&= \delta_{k,\ell} \int_\Domain \pdesolutionoperator^h(\mathbf{e}_m)(\tk[\ell], \mathbf{x}) [\nabla_{\mathbf{y}} \MeasKern(\mathbf{x}, \pk) ]_i ~  d\mathbf{x}.
\end{align}
for all $\ell \in \{1,\dots, \dimTime\}$, $m\in\{1,\dots, \dimParameter\}$, and $i\in \{1,\dots,\dimDomain\}$.
\end{corollary}
\begin{proof}
The proof is provided in \ref{sec:proofs}.
\end{proof}

Global, continuous differentiability of $\MeasKern$ is a strong assumption; it can be weakened if the parameter-to-state map $\pdesolutionoperator^h$ is sufficiently smooth. 

\begin{corollary}\label{thm:gradient:3}
Suppose $\ObsPath^h \subset \Domain$ and $\pdesolutionoperator^h : \parameterspace \rightarrow \mathcal{C}([0, \finalt], H^1(\Domain))$.
Suppose further that $\MeasKern$ is of the form
\begin{align}\label{eq:measkern:alternative}
    \MeasKern(\textbf{x}, \textbf{y}) = c_{\phi} \mathbbm{1}_{B_r(\mathbf{y})} \phi(\mathbf{x}-\mathbf{y})
\end{align}
where $c_{\phi} > 0$ is a scaling constant, $B_r(\mathbf{y})$ the ball of radius $r >0$ around the point $\mathbf{y}$, $\mathbbm{1}_{B_r(\mathbf{y})}$ is the indicator function, and $\phi \in C(\overline{B_r(\mathbf{0})})$ is continuous on the closed ball $\overline{B_r(\mathbf{0})} \subset \R^{\dimDomain}$.
If $r > 0$ is small enough such that $B_r(\pk^h) \subset \Domain$ for all $k \in \{1, \dots, \dimTime\}$, then the parameter-to-observable matrix $\paratoobsDiscrete$ is continuously differentiable with respect to $\pk^h$ for all $k \in \{1, \dots, \dimTime\}$ with
    \begin{align}\label{eq:deriv:G:2}
        \left[\frac{d \paratoobsDiscrete}{d \pki} \right]_{\ell, m} 
    &= \delta_{k, \ell} \int_{B_r(0)} \frac{\partial \pdesolutionoperator^h(\m)}{\partial \mathbf{x}_i}(\tk, \mathbf{x} + \pk) \phi(\mathbf{x})d\mathbf{x}
    \end{align}
\end{corollary}
\begin{proof}
The proof is provided in \ref{sec:proofs}.
\end{proof}

Naturally, the numerical evaluation of the integrals in \eqref{eq:deriv:G:1} (or \eqref{eq:deriv:G:2}) depends on the chosen discretization, and can hence become computationally expensive. 
In addition, the computation of $\nabla_{\pk} \ObjAOED(\CovPostDiscrete)$ or $\nabla_{\pk} \ObjDOED(\CovPostDiscrete)$ through \eqref{eq:derivative:A:discrete} (for A-\ac{oed}) or \eqref{eq:derivative:D:discrete} (for D-\ac{oed})---even after computing $\frac{d \paratoobsDiscrete}{d \pki}$---further involves the repeated evaluation of the involved matrices' actions to compute the trace.
While some of these computations are shared among all time steps $\tk$, the action of the derivative $\frac{d \paratoobsDiscrete}{d \pki}$ differs, resulting in a total of $\dimTime \dimParameter$ different computations of matrix traces.
To reduce computational costs, we suggest warm-starting the optimization problem with coarser discretizations and pointwise measurements to identify suitable initial guesses before changing to the intended discretization and measurement kernel function $\MeasKern$.
This procedure is demonstrated in Section \ref{sec:results:transient:warmstart}.

\section{Numerical Results}\label{sec:numerics}

In this section, we demonstrate our \ac{oed} algorithm for moving sensors.
The forward model is described in Section \ref{sec:results:forward}.
In Section \ref{sec:results:transient:2D}, we visualize the A- and D-\ac{oed} utility functions w.r.t. the initial heading and (constant) angular velocity control parameters and choose local minima as initial guesses for the optimization.
In Section \ref{sec:results:transient:oneshot}, we present the results of the optimization.
In Section \ref{sec:results:transient:warmstart}, we demonstrate how our optimization algorithm can be warm-started to refine the discretization of the measurements.
Finally, in Section \ref{sec:results:transient:highdim}, we increase the parameter dimension to $\dimParameter = 25$ to show the scaling of our algorithm.

\subsection{Forward model}\label{sec:results:forward}

We consider the spread of a pollutant in a domain $\Domain \subset (0,1) \times (0,1)$ modeled by the convection-diffusion equation
\begin{align}\label{eq:transient}
    \frac{d\u}{dt} &= 10^{-3} \Delta \u - \mathbf{v} \cdot \nabla \u & \text{in } \Domain,
\end{align}
with initial condition $\u(0) = \u_0$ and Neumann boundary conditions $\nabla \u \cdot \mathbf{n} = 0$ on $\partial \Omega$ with $\mathbf{n} : \partial \Domain \rightarrow \mathbb{R}^2$ being the outside-pointing unit normal.
The domain $\Domain$ and the velocity field $\mathbf{v} : \Domain \rightarrow \mathbb{R}^2$ are shown in Figure \ref{fig:velocity-field}.
The velocity field $\mathbf{v}$, together with the pressure $p : \Omega \rightarrow \mathbb{R}$, is described through the steady-state Navier-Stokes equations
\begin{align}\label{eq:Navier-Stokes}
    \frac{1}{500} \Delta \mathbf{v} + \nabla p + \mathbf{v} \cdot \nabla \mathbf{v} &= 0 & \text{in } \Domain, \\
    \nabla \cdot \mathbf{v} &= 0 & \text{in } \Domain,
\end{align}
with imposed Dirichlet boundary conditions: on the left wall, the velocity field is oriented in $\mathbf{x}_2$-direction with $\mathbf{v}(0,\mathbf{x}_2) = (0, 1)^{\top}$ for all $\mathbf{x}_2\in (0,1)$; on the right wall, it is oriented in the reverse direction with $\mathbf{v}(1,\mathbf{x}_2) = (0, -1)^{\top}$ for all $\mathbf{x}_2\in (0,1)$; on the remaining boundaries it is zero, i.e., $\mathbf{v}(\mathbf{x}) = (0,0)^{\top}$ for $\mathbf{x} = (\mathbf{x}_1, \mathbf{x}_2)^{\top}\in \partial \Domain$, $\mathbf{x}_1\notin \{0, 1\}$.

\begin{figure}
    \centering
    \includegraphics[width=0.5\linewidth]{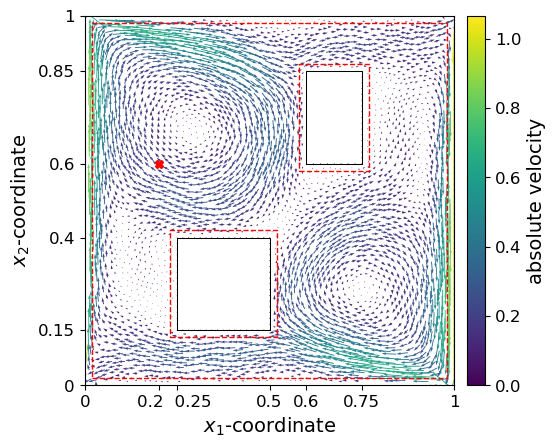}
    \caption{Computational domain $\Domain$ and velocity field $\mathbf{v}$. 
    Red dashed lines demarcate the the permitted flight area.
    The red x-marker is the starting position for the moving sensor at time $t=0$.
    }
    \label{fig:velocity-field}
\end{figure}

We discretize the convection-diffusion equation \eqref{eq:transient} in space using cubic finite elements with 32,096 total degrees of freedom, and solve it for $t \in [0, 5]$ using Crank-Nicolson time stepping with uniform time steps $t_k = k \Delta t$, $k=0, \dots, 5001$ and step size $\Delta t = 0.001$.
Employing a sparse LU-factorization of the discretized convection-diffusion operator, a forward solve with a given initial condition $\u(0) = \u_0$ takes approximately 70 s.
Because there is no forcing term in \eqref{eq:transient}, the Crank-Nicolson time stepping is equivalent to the continuous Galerkin space-time finite element method with piecewise linear basis functions in time (see \cite[Chapter 12]{Thomee2007}), which allows us to evaluate $\u(t) = \frac{t-s_{i}}{\Delta s}\u(s_{i+1}) + (1-\frac{t-s_{i}}{\Delta s}) u(s_{i})$ for any $t \in (t_k, t_{k+1})$ between any two time steps $t_k$ and $t_{k+1}$.

To formulate our \ac{oed} problem, we consider the following scenario in Sections \ref{sec:results:transient:2D} -- \ref{sec:results:transient:warmstart}:
It is known that, at time $t=0$, a pollutant is released around two positions $\mathbf{c}_1 := (0.1, 0.9)$ and $\mathbf{c}_2 := (0.7, 0.1)$
such that the initial state $\u_0$ takes the form
\begin{align}
    \u_0(\mathbf{x}; \m) = \m_1 \min\{\exp(-100\|\mathbf{x}-\mathbf{c}_1\|^2), \frac12\} + \m_2 \min\{\exp(-100\|\mathbf{x}-\mathbf{c}_2\|^2), \frac12\}.
\end{align}
However, the amount of released pollutant is unknown, and we model the scaling vector $\m = (\m_1, \m_2)^{\top} \in \mathbb{R}^2$ as a random variable with prior $\mathcal{N}((1,1)^{\top}, \Identity[2])$.
A negative entry in $\m$ and consequently negative concentrations in the state $\u(\m)$ are interpreted as the release of a counter-agent that neutralizes itself and the pollutant in a 1:1 ratio.
For reference, the evolution of the state $\u$ starting from $\u(0)=\u_0(\,\cdot\,; \m)$ is shown in Figure \ref{fig:transient:sol} for $\m = \mbasis[1]$ (top) and $\m = \mbasis[2]$ (bottom).

\begin{figure}
    \centering
    \includegraphics[width=\linewidth]{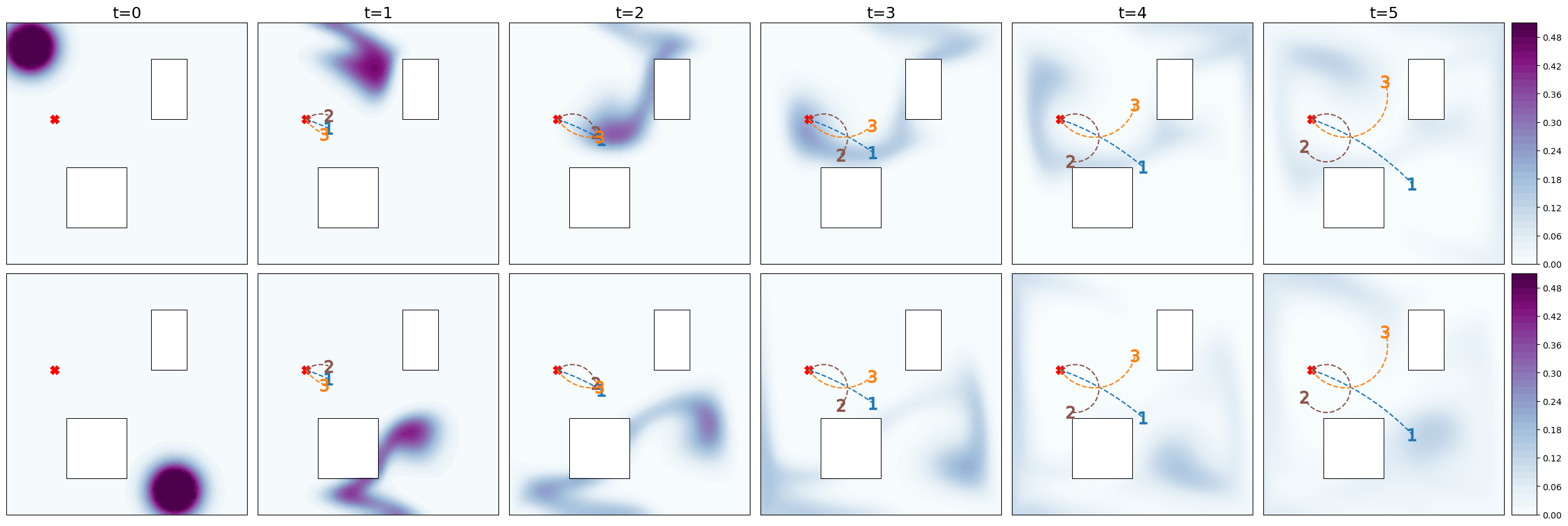}
    \caption{Solution of the \ac{pde} for two initial conditions (top: $\m = \mbasis[1]$, bottom: $\m = \mbasis[2]$) at time steps $t=0, 1, \dots, 5$.
    The numbered paths are local minima explained in section \ref{sec:results:transient:2D} and
    marked in Figure \ref{fig:utility2D:D}.
    }
    \label{fig:transient:sol}
\end{figure}

The moving sensor starts at time $t=0$ from the position $(0.2, 0.6)^{\top}$, marked with a red cross in Figures \ref{fig:velocity-field} and \ref{fig:transient:sol}.
The sensor measures the concentration of the pollutant pointwise\footnote{We use pointwise measurements as a first approximation to generate an initial guess for the path $\ObsPath$. We switch to measurements of the form \eqref{eq:data:uniform} in Section \ref{sec:results:transient:warmstart}.}
at its own position $\ObsPath(t) \in \Domain$, with noise distributed as described in Example \ref{ex:noise:2}.
For the three paths marked in Figure \ref{fig:transient:sol}, the signal-to-noise ratio $10 \log_{10}(\frac{\mathbb{E}_{\m}[\d(\m)(t)^2]}{\mathbb{E}_{\eta}[\eta(t)^2]})$ with this noise model varies from approximately 12 dB to 16 dB (estimated from 10,000 samples of the prior and the noise model).
The path $\ObsPath$ of the moving sensor is modeled through the \ac{oed} \eqref{ex:ode} with constant velocity $v \in [0.05, 0.2]$, angular velocity $\omega(t) \in [-2, 2]$, and initial heading $\theta(0) \in [-\pi, \pi]$ to be determined in the optimization.
The \ac{ode} is discretized initially on the time interval $[0, 5]$ with time step size $\dt = 10^{-2}$; convergence for $\dt \rightarrow 0$ is investigated in Section \ref{sec:results:transient:warmstart}.
We require the sensor to keep a safety distance of at least 0.02 to the edge $\partial \Domain$ of the domain, as indicated by the red dashed lines in Figure \ref{fig:velocity-field}.

\subsection{Initial heading vs angular velocity}\label{sec:results:transient:2D}

\begin{figure}
    \centering
    \begin{subfigure}[t]{0.495\textwidth}
    \includegraphics[width=\linewidth]{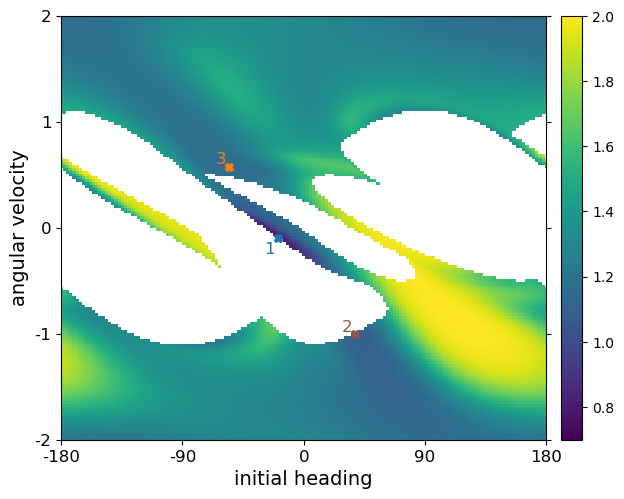}
    \caption{A-OED utility}\label{fig:utility2D:A}
    \end{subfigure}
    \begin{subfigure}[t]{0.495\textwidth}
    \includegraphics[width=\linewidth]{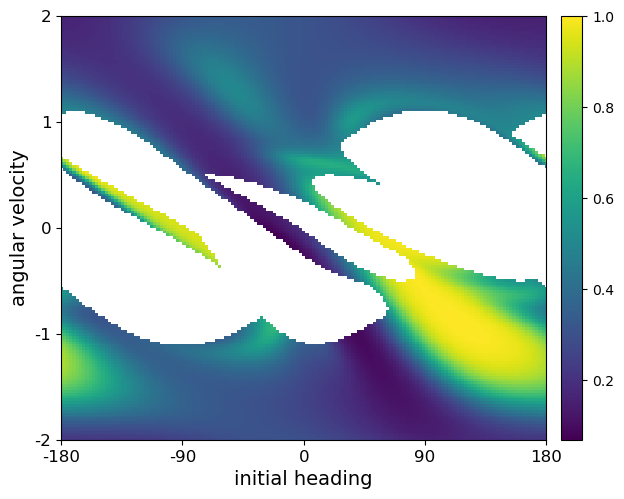}
    \caption{D-OED utility}\label{fig:utility2D:D}
    \end{subfigure}
    \caption{
    A- and D-\ac{oed} utility values for paths characterized by varying initial heading and angular velocity (constant in time). 
    Initial position and velocity $\velocity(t) \equiv 0.1$ are fixed. 
    White areas correspond to paths outside the admissible domain.
    The colored and numbered x-markers indicate selected local minima of the A-\ac{oed} utility function used to initialize the optimization.
    The corresponding paths are shown in Figure \ref{fig:transient:sol}.
    }
    \label{fig:utility2D}
\end{figure}

As a first step in optimizing the path $\ObsPath$ of the moving sensor, we consider curves only with constant velocity $\velocity(t) \equiv 0.1$ and a constant angular velocity $\angular \equiv \angular_0 \in [-2, 2]$.
Figure \ref{fig:utility2D} shows the A- and D-OED utility functions for varying $\angular_0 \in [-2, 2]$ and initial heading $\heading(0) \in [-\pi, \pi]$.
The figure clearly shows that the admissible domain $\admissible$ for the control parameterization is non-trivial and that the A- and D-\ac{oed} utility functions are non-convex, exhibiting many local minima.

\begin{figure}
    \centering
    \includegraphics[width=\linewidth]{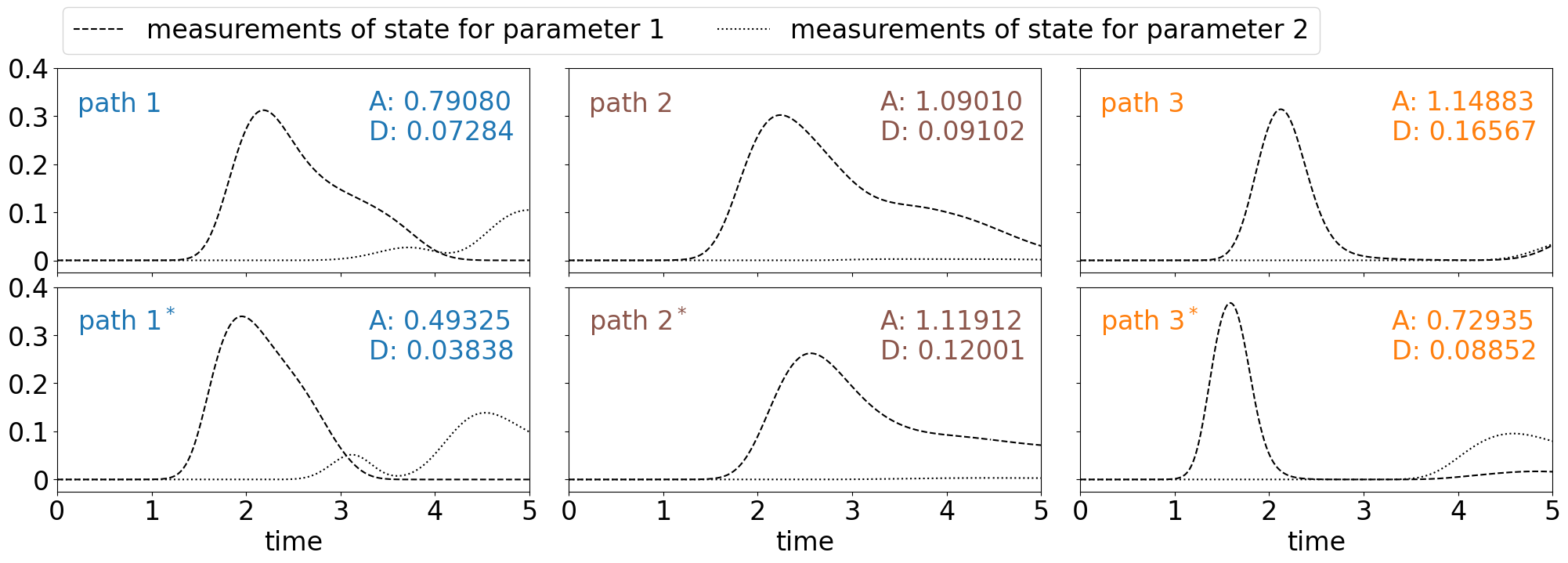}
    \caption{Measurements of the states $\u(\mbasis[1])$ (dashed) and $\u(\mbasis[2])$ (dotted) along the initial path guesses (top) and the optimized paths (bottom).}
    \label{fig:measurements-by-parameter}
\end{figure}

For further investigation, we choose the three local minima of $\ObjAOED$ indicated in Figure \ref{fig:utility2D:A}.
Their associated paths are shown in Figure \ref{fig:transient:sol}.
All three paths lead the moving sensor into a high-concentration pollutant cloud originating from the first, top left source around $t=2$. 
Afterwards, on path 1, the moving sensor passes between the two houses and enters the bottom right of the domain, ending, at final time $t=5$, in a cloud of pollutant originating from the second source on the bottom right.
Consequently, the sensor obtains large and distinct measurements of each pollutant source, leading to small A- and D-\ac{oed} utility values (see Figure \ref{fig:measurements-by-parameter}, top left).
In contrast, on path 2, the sensor moves clockwise in a small circle, well tracking the first cloud of pollutant while only encountering negligible measurements ($\le 0.003$) from the second source.
As will be discussed in Section \ref{sec:results:transient:oneshot} below, this missing awareness to the second pollutant source makes path 2 a difficult initial guess for the optimization algorithm.
In comparison, path 3 moves counterclockwise through the first pollutant cloud in order to then re-enter a second pollutant cloud containing pollutant from both sources.
The optimization algorithm can then steer the moving sensor to further increase measurements of $\u(\mbasis[2])$ (see Figure \ref{fig:measurements-by-parameter}, top right), making path 3 a more valuable initial guess than path 2. 

\subsection{Optimization from initial guess}\label{sec:results:transient:oneshot}

We use paths 1, 2, 3 as initial guesses for the optimization.
Specifically, we use the \ac{ipopt} (\cite{wachter2006implementation}) python package to solve \eqref{eq:FiniteDimOED} with the A-\ac{oed} utility function $\ObjOED=\ObjAOED$ and regularization weight $\gamma=0.1$.
The regularization weight was chosen such that the percentage of the regularization term in the total cost function is below $10\%$ for all three initial guesses.
The optimization is constrained to enforce the initial sensor position $\ObsPath(0) = (0.2, 0.6)^{\top}$, and for the path to remain a safety distance of $0.02$ away from the domain boundary (see red dashed lines in Figure \ref{fig:velocity-field}).
Furthermore, we impose that the velocity $\velocity$ is constant with initial value $0.05 \le \velocity(0) \le 0.2$.
We use the Quasi-Newton optimization with the IPOPT default options; in particular, the optimization terminates when the \ac{nlp} error is either below $10^{-8}$ (``optimal''), or remained below $10^{-6}$ for 15 consecutive iterations (``acceptable'').
For the time step size $\dt = 10^{-2}$, the optimization problem has 2,508 optimization variables, 2,005 equality and 1,002 inequality constraints.

\begin{table}[t]
    \centering
    \begin{tabular}{ll|rrr}
         & initial guess & path 1 & path 2 & path 3 \\
         \hline
        cost & at initialization & 0.854 & 1.216 & 1.232 \\
        & after optimization & 0.611 & 1.149 & 0.974 \\
        & change (absolute) & -0.243 & -0.067 & -0.258 \\
        & change (relative) & $\times$0.715 & $\times$0.945 & $\times$0.791 \\
        \hline
        A-\ac{oed} & initial & 0.791 & 1.090 & 1.149 \\
        & after & 0.493 & 1.119 & 0.729 \\
        & change (absolute) & -0.298 & +0.029 & -0.419 \\
        & change (relative) & $\times$0.624 & $\times$1.027 & $\times$0.635 \\
        \hline
        D-\ac{oed} & initial & 0.073 & 0.091 & 0.166 \\
        & after & 0.038 & 0.120 & 0.089 \\
        & change (absolute) & -0.034 & +0.029 & -0.077 \\
        & change (relative) & $\times$0.527 & $\times$1.319 & $\times$0.534 \\
        \hline
        stats & no. iterations & 286 & 65 & 886 \\
        & no. fct. evals & 770 & 74 & 5161 \\
        & no. gradient evals & 287 & 66 & 905 \\
        \hline
        finish & compute time [s] & 49.8 & 9.8 & 233.1 \\
        & termination & optimal & optimal & accept. \\
        & constr. violation & 4.14e-11  & 1.85e-14 & 2.97e-11 \\
        & \acs{nlp} error & 2.07e-9 & 2.34e-9 & 5.84e-7 \\
    \end{tabular}
    \caption{Optimization results (cost function \eqref{eq:FiniteDimOED:concrete:a} with $\ObjOED = \ObjAOED$), starting from the local minima selected in Figure \ref{fig:utility2D:A}.
    }
    \label{tab:optimization:one}
\end{table}

The key output data for the optimization starting from the initialization paths 1, 2, 3 are listed in Table \ref{tab:optimization:one}.
In all three cases, the optimization terminated successfully, and the cost functions were reduced.
For paths 1 and 3, the the A-\ac{oed} utility decreased to below $64\%$ (lower is better).
The D-\ac{oed} utility values, although not directly part of the cost function, also decreased, illustrating the connection of $\ObjAOED$ and $\ObjDOED$ through the eigenvalues of the posterior covariance matrix.
In contrast, for path 2, only the cost function decreased while both utility values increased, i.e., the optimization algorithm primarily decreased the regularization term.
This behavior indicates that path 2 was already close to a local minimum of $\ObjAOED$.
This explanation is further supported by the comparatively fast convergence of the optimization algorithm for the initial guess 2 (65 iterations for initial guess 2 compared to 286 iterations for initial guess 1).

It is notable that for path 3, the optimization took more than three times as many iterations as for path 1, and only converged to acceptable precision.
The slower convergence can be explained by the shape of the identified path, shown in Figure \eqref{fig:optimization:paths:1shot:A}:
The optimized sensor turns quickly and steers counter-clockwise along the outer edge of the admissible domain until, at the final time $\finalt = 5$, it is on the boundary $\partial \admissibleDomain$.
This maneuver exploits the clockwise, high velocity circulation along the outer boundary of the domain in order to obtain measurements of pollutant originating from the second source at the bottom right.
However, the closeness to the edge of the admissible domain causes convergence problems in the constraints, as reflected by the increased runtime.

\begin{figure}
    \centering
    \includegraphics[width=\linewidth]{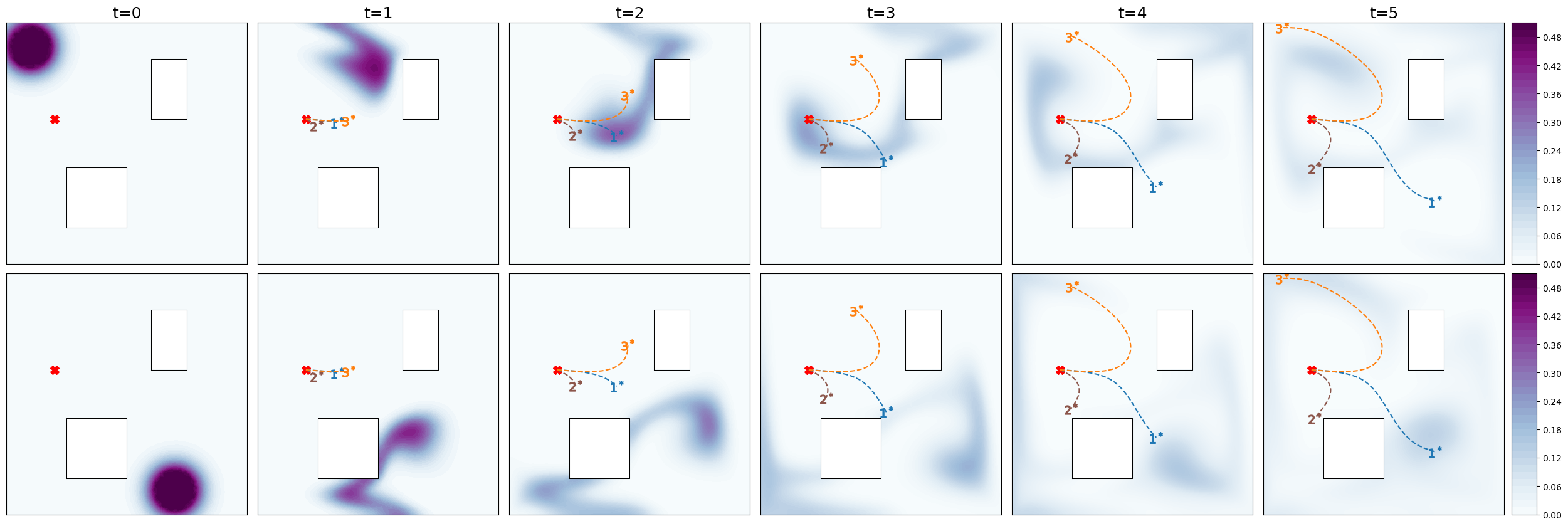}
    \caption{Solution of the \ac{pde} for two initial conditions at time steps $t=0, 1, \dots, 5$.
    The numbered paths are the solutions of the optimization problem for three initial guesses.}
    \label{fig:optimization:paths:1shot:A}
\end{figure}

The optimized paths are shown in Figure \ref{fig:optimization:paths:1shot:A}.
For path $1^*$, the velocity has increased from $\velocity \equiv 0.1$ at the initial guess to $\velocity \equiv 0.131$, and the angular velocity has been adjusted locally for the moving sensor to pass directly through high-concentration areas (see Figure \ref{fig:optimization:paths:1shot:A} for $t=2$ and $t=5$).
The velocity of path $3^*$ has increased even further to $\velocity \equiv 0.178$, allowing the moving sensor to reach and move alongside the upper domain boundary.
For both optimized paths $1^*$ and $3^*$, the changes compared to the initial guesses are driven to increase measurements, especially of pollutant originating from the second, bottom right source.
This is also evident when comparing the top and bottom rows of Figure \ref{fig:measurements-by-parameter}:
The optimized moving sensors measure pollutant originating from either source earlier and at a higher concentration. 
For path $2^*$ where the optimization decreased the regularization term rather than the A-\ac{oed} utility value, the velocity has decreased to $\velocity \equiv 0.052$, making the sensor's path (Figure \ref{fig:optimization:paths:1shot:A}) rather short.
However, there is little change in the measurements of $\u(\mbasis[1])$ (Figure \ref{fig:measurements-by-parameter}, middle column), such that the A- and D-\ac{oed} values are only slightly higher than at the initial guess.

\subsection{Warm start and convergence}
\label{sec:results:transient:warmstart}

\begin{figure}
    \centering
    \includegraphics[width=\linewidth]{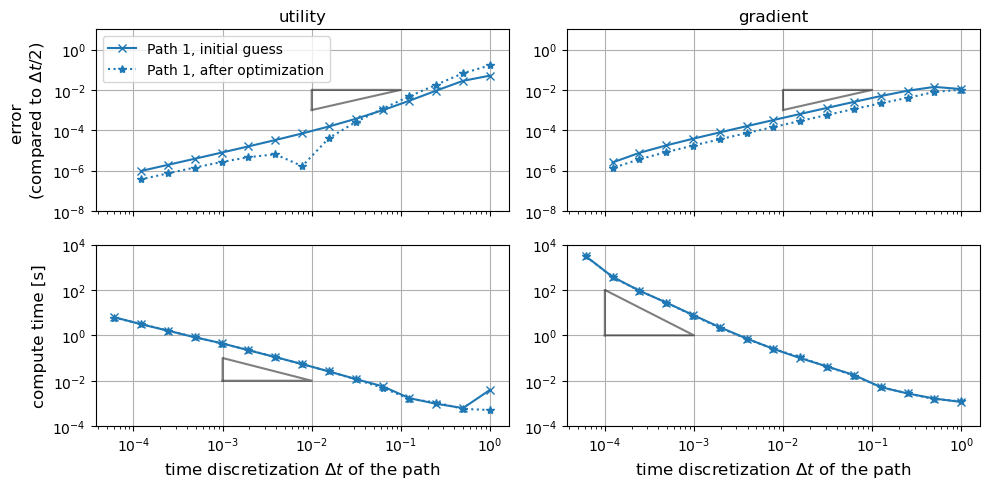}
    \caption{Convergence error (top) and compute time (bottom) of the A-\ac{oed} utility $\ObjAOED$ (left) and its gradient $\nabla \ObjAOED$ at the initial path 1 before and after optimization. Lines in the bottom plots are superimposed.}
    \label{fig:results:convergence}
\end{figure}

Next, we investigate the convergence of the A-\ac{oed} utility and its gradient under refinement of the sensor's time discretization $\dt$ using the control parameterizations of the initial path $1$ and its optimized counterpart $1^*$.
Note that we do not change the time discretization $\Delta s = 10^{-3}$ of the forward finite element model as the Continuous Galerkin time stepping scheme allows us to evaluate $\u(t)$ for any $t \in [0, \finalt]$.
To change the time discretization $\Delta t$ of the paths, we interpolate the paths' velocity and angular velocity control parameters onto the new grid points, and then solve the \ac{ode} anew to obtain the corresponding paths.
Consequently, changing $\dt$ not only effects at which times measurements $\d(t)$ are evaluated in the Bayesian inverse problem and how they are weighted through the inverse noise covariance matrix $\CovNoise^{-1}$, but also at which spatial positions the measurements are taken.
Still, we observe linear convergence as $\dt$ decreases for both $\ObjAOED(\ObsPath)$ and $\nabla \ObjAOED(\ObsPath)$.
The results are shown in Figure \ref{fig:results:convergence} (top row).
We only see one outlier, specifically for path $1^*$ at $\dt =7.81 \cdot 10^{-3}$.
A reason for this outlier could be that path $1^*$ is optimized for $\dt=10^{-2}$, causing the A-\ac{oed} utility value to be particularly small for similar $\dt$, while for much finer or coarser $\dt$ the \ac{oed} converges further away from the minimum.
This behavior indicates that the \ac{oed} optimization problem \eqref{eq:FiniteDimOED} should indeed be solved for the time discretization $\dt$ at which the Baysian posterior \eqref{eq:postMeasure} will be computed after the observational data $\dobs$ is obtained.
However, as shown in Figure \ref{fig:results:convergence} (bottom row), the computation time for the A-\ac{oed} utility value increases linearly as $\dt$ decreases, and increases almost quadratically for evaluating the corresponding gradient.
This makes the optimization problem increasingly expensive to solve:
Assuming the optimization would need the same number of function and gradient evaluations as for $\dt=10^{-2}$ (770 and 287, respectively), then the compute time would increase from 49.8 s for $\dt=10^{-2}$ to above 13 min for $\dt = 10^{-3}$ and above 29 h for $\dt = 10^{-4}$.
Instead, for small $\dt$, we warm-start the optimization with the optimal paths obtained for coarser $\dt$.

To demonstrate warm-starting, we interpolate the control parameters and the Lagrangian multipliers of the optimized path $1^*$ for $\dt = 10^{-2}$ from Section \ref{sec:results:transient:oneshot} onto a finer time discretization with halved time step size $\dt = \frac12 \cdot 10^{-2}$; we use them as initial guesses to solve refined optimization problem.
We repeat this process six times until $\dt = (\frac{1}{2})^{6} \cdot 10^{-2} \approx 1.5 \cdot 10^{-4}$; the results are reported in Table \ref{tab:results:warm}:
In each iteration of halving $\dt$, the number of iterations within the optimization algorithm and consequently the number of function and gradient evaluations decreases.
Finally, when warm-starting for $\dt = (\frac{1}{2})^{6} \cdot 10^{-2}$, the optimum is already reached within one iteration of the optimization algorithm, requiring only two function and gradient evaluations.
In combination with the decreasing initial \ac{nlp} error for each warm-started optimization, this behavior indicates that we are indeed converging to a local minimum of $\ObjAOED(\CovPost)$ for continuous time. 
Moreover, it shows that the discretization of the Baysian inverse problem described in Section \ref{sec:discretization} remains indeed consistent for $\dt \rightarrow 0$, i.e., evaluating observational data $\dobs$ at arbitrary many time steps does not result in an arbitrarily large reduction in uncertainty (as can be the case for uncorrelated measurements) but converges.

\afterpage{%
    \begin{landscape}
        \begin{table}[]
    \centering
    \begin{tabular}{l|rr|rrrrrr}
        iteration & grid & 0 & 1 & 2 & 3 & 4 & 5 & 6 \\
        $100\dt$ & $1$ & $1$ & $2^{-1}$ & $2^{-2}$ & $2^{-3}$ & $2^{-4}$ & $2^{-5}$ & $2^{-6}$ \\
        \hline
        initial \acs{nlp} error & 
        --- & --- & 4.69e-3 & 2.33e-3 & 1.17e-3 & 5.82e-4 & 2.90e-4 & 1.45e-4
        \\
        final \acs{nlp} error & 
        --- & 2.07e-9 & 8.47e-7 & 4.17e-9 & 4.69e-9 & 6.87e-9 & 6.34e-9 & 5.31e-9
        \\
        termination & 
        --- & opt. & acc. & opt. & opt. & opt. & opt. & opt.
        \\
        \hline 
        no. iterations & 
        --- & 286 & 142 & 76 & 44 & 12 & 3 & 1
        \\
        no. fct. evals & 
        22650 & 770 & 540 & 244 & 214 & 14 & 4 & 2
        \\
        no. grad. evals & 
        --- & 287 & 144 & 77 & 45 & 13 & 4 & 2
        \\
        \hline
        comp. time [min] & 
        5.97 & 0.83 & 1.17 & 2.29 & 4.61 & 4.03 & 4.56 & 10.80
        \\
        cuml. time [min] & 
        5.97 & 6.80 & 7.97 & 10.26 & 14.87 & 18.90 & 23.46 & 34.26
        \\
        \hline
        cost function val.&
        0.853980 & 0.611383 & 0.611261 & 0.611203 & 0.611175 & 0.611161 & 0.611154 & 0.611151
        \\
        change & --- & -2.43e-1 & -1.22e-4 & -5.77e-5 & 
        -2.80e-5 & -1.38e-5 & -6.85e-6 & 
        -3.40e-6 \\
        \hline 
        A-\ac{oed} &
        0.790802 & 0.493247 & 0.493238 & 0.493196 & 0.493187 & 0.493179 & 0.493183 & 0.493183
        \\
        change & --- & -2.98e-1 & -9.74e-6 & -4.17e-5 & 
        -8.76e-6 & -7.80e-6 & 3.27e-6 & 
         3.81e-9 \\
         \hline 
        D-\ac{oed} &
        0.072840 & 0.038384 & 0.038384 & 0.038376 & 0.038375 & 0.038375  & 0.038375 & 0.038378
        \\
        change & --- & -3.45e-2 & -3.39e-7 & -7.62e-6 & 
        -1.24e-6 & -5.47e-7 & 5.25e-7 & 
         2.94e-6
    \end{tabular}
    \caption{Warmstart optimization results (cost function \eqref{eq:FiniteDimOED} with $\ObjOED=\ObjAOED$ and $\gamma=0.1$) with decreasing time discretization $\dt$.
    The second column is the grid search for path $1$ (Section \ref{sec:results:transient:2D}), the third column is the optimization for path $1^*$ (Section \ref{sec:results:transient:oneshot})}
    \label{tab:results:warm}
\end{table}
    \end{landscape}
    \clearpage
}

Up until now, we have reported results for pointwise measurements, as these are comparatively cheap to compute and therefore favorable to reduce the compute time of the optimization algorithm.
However, an identified optimal path for pointwise measurements may not be optimal for other measurement types of the general form \eqref{eq:data:general}.
Yet, it is an educated initial guess if the kernel function $\MeasKern$ weights the position of the moving sensor strongly.
We illustrate this for measurements of the form \eqref{eq:data:uniform}, i.e., measurements $\d(t;\ObsPath, \u) := \mathbb{E}_{\mu}[\u(t, \cdot)]$ where $\mu = \mu(t, \ObsPath) = \mathcal{U}(B_r(\ObsPath(t)))$ is the uniform distribution on a ball with radius $r$ centered around the position of the moving sensor.
For any time $t \in [0, 5]$, we compute these measurements by interpolating $\u(t, \cdot)$ onto a fine mesh (568 nodes) of $B_r(\ObsPath(t))$ before evaluating the integral.
We estimate that this procedure introduces an approximation error of approximately $~10^{-6}$ to the A-\ac{oed} utility function, and increases its compute time for a given path $\ObsPath$ to approximately 25 s per evaluation. 
Figure \ref{fig:convolution:A} shows the A-\ac{oed} utility of path $1^*$ for radii $r \in (0.001, 0.0199)$.
We avoid radii $r \ge 0.02$ as $B_r(\ObsPath(t))$ may then extend past the safety distance and intersect $\partial \Domain$.

\begin{figure}
    \centering
    \includegraphics[width=0.8\linewidth]{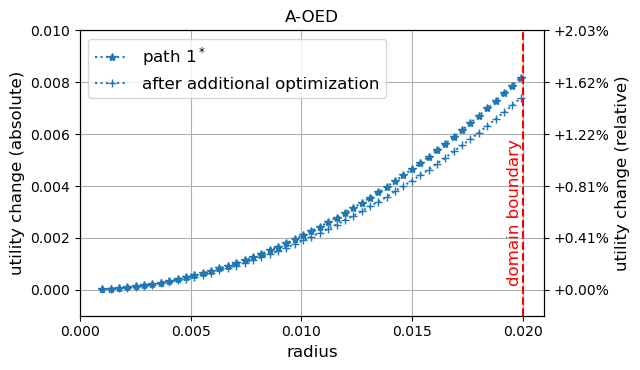}
    \caption{Changes in the A-\ac{oed} utility value when using averaged measurements of the form \eqref{eq:data:uniform} instead of pointwise measurements for different radii.}
    \label{fig:convolution:A}
\end{figure}

We observe in Figure \ref{fig:convolution:A} that $\ObjAOED(\ObsPath)$ increases as $B_r(\ObsPath(t))$ becomes larger.
This is to be expected as the integral in \eqref{eq:data:uniform} averages local information.
However, the changes in utility compared to pointwise measurements are slim:
For $r=0.0199$, $\ObjAOED(\ObsPath) = 0.501$, which is a factor $\times 1.0162$ increase compared to pointwise measurements.
Moreover, path $1^*$ is almost optimal for the considered radii:
Warmstarting initially with path $1^*$ and then with the obtained minima, the optimization algorithm successfully converged for all tested radii (32 times to a local optimum, 18 times to an acceptable level), terminating within 32 iterations for most\footnote{In four cases the optimization took longer: 256 iterations for $r=0.0097$, 196 iterations for $r=0.01423$, 149 iterations for $r=0.0146$, and 219 iterations for $r=0.0157$.}, with fewer iterations primarily for the smaller radii.
Yet, this additional optimization changes $\ObjAOED(\ObsPath)$ only slightly, with the largest changes observed for the largest radii, confirming the intuition that pointwise measurements can be used to warmstart the optimization.

\subsection{High-dimensional parameter space}\label{sec:results:transient:highdim}

\begin{figure}
    \centering
    \includegraphics[width=\linewidth]{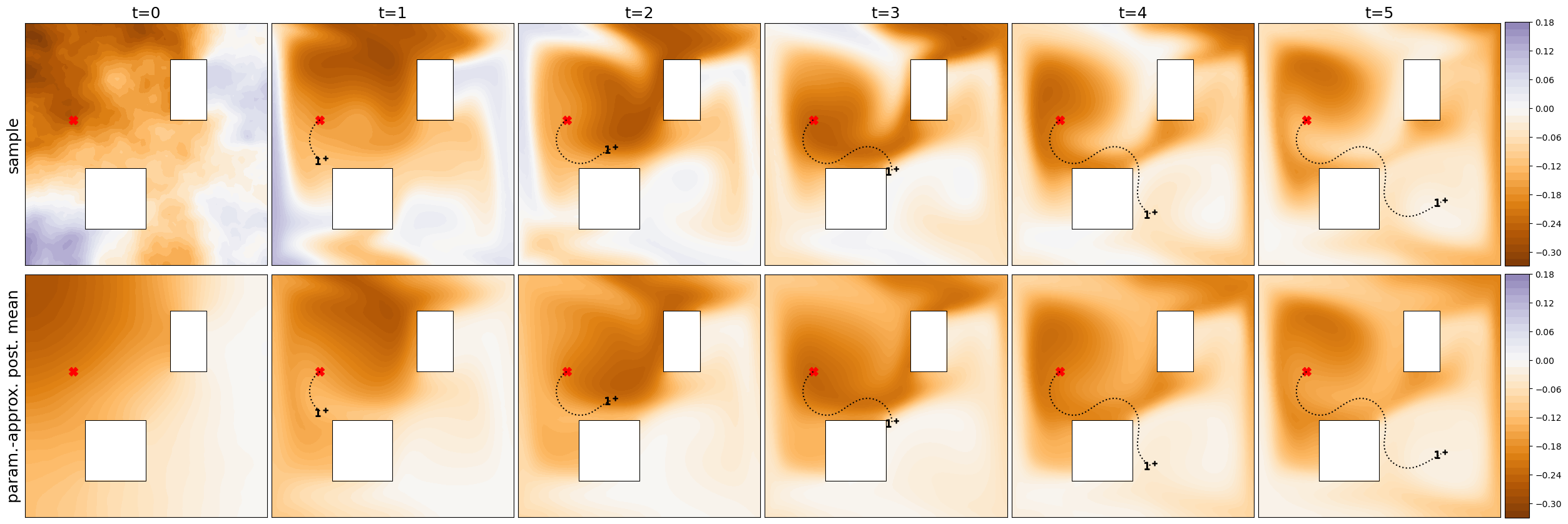}
    \caption{Top: Sample of the prior ($t=0$) and its evolution in time. Bottom: Parameter-reduced posterior mean ($t=0$) and its evolution in time.}
    \label{fig:highdim:sample}
\end{figure}

Finally, after establishing the convergence behavior of our \ac{oed} formulation and the optimization algorithm on the identification of $\dimParameter = 2$ scaling values, we consider the Bayesian inverse problem of inferring the initial condition $u_0 \in L^2(\Omega)$ as a \textit{field}.
Specifically, we model $u_0 \sim \mathcal{N}(0, \mathcal{C}_{\rm{pr}})$ with $\mathcal{C}_{\rm{pr}}$ chosen as the squared (weak) solution operator of the \ac{pde}
\begin{align}
    -\Delta u + 8 u &= 0 \text{ in $\Domain$, } &\nabla u \cdot \mathbf{n} &= 0 \text{ on $\partial \Domain$}.
\end{align}
This definition of the prior $\prior$ is discussed extensively in \cite{BuiThanh2013} and guarantees that $\mathcal{C}_{\rm{pr}}$ is indeed trace class. 
A sample of the prior and its corresponding state evolved through the \ac{pde} \eqref{eq:transient} are shown in Figure \ref{fig:highdim:sample} (top row).

\begin{figure}
    \centering
    \includegraphics[width=\textwidth]{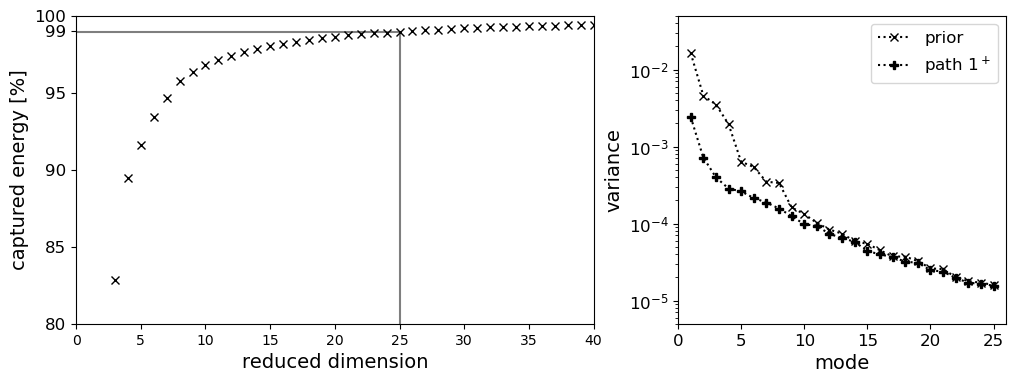}
    \caption{Left: Percentage of captured energy (in percent) for varying dimensions of the reduced space.
    Right: Eigenvalues of the parameter-reduced prior and posterior covariance matrices.
    }
    \label{fig:highdim:energy}
\end{figure}

Because we have restricted the exposition in our manuscript to parameters $\m \in \parameterspace$ with $\dimParameter \in \mathbb{N}$ and the efficient, scalable extension to parameter fields is part of future work, we adopt the parameter-approximated posterior from \cite{cui2016scalable}.
To this end, we decompose $u_0 = u_0^{\dimParameter} + u_0^{\perp}$ where $u_0^{\dimParameter} \in \mathcal{V}$ is the orthogonal projection of $u_0$ onto an $\dimParameter$-dimensional subspace $\mathcal{V} \subset L^2(\Omega)$ and $u_0^{\perp} \in \mathcal{V}^{\perp}$ is its orthogonal complement.
We choose $\mathcal{V}$ through a principal component analysis with 1,000 samples of the prior; for $\dimParameter = \dim \mathcal{V} = 25$, our reduced space captures 99.01 \% of the sample energy (see Figure \ref{fig:highdim:energy}, left).
Noting that $u_0^{\dimParameter} \in \mathcal{V}$ can be uniquely identified with a vector $\m \in \mathbb{R}^{\dimParameter}$ such that $u_0^{\dimParameter} = \sum_{i=1}^{\dimParameter} \m_i v_i$ where $v_1, \dots, v_{\dimParameter}$ are the basis functions of $\mathcal{V}$, the parameter-approximated posterior is defined as
\begin{align}
\widetilde{\pi}(u_0 | \dobs) \propto \widehat{\pi}(\m | \dobs) \pi_0(u_0^{\perp}) \propto \exp \left( -\| \OpCovNoise^{-1/2}(\paratoobs \m - \dobs) \|_{\Ltwotime}^2\right) \pi_0(u_0)
\end{align}
with the parameter-reduced posterior defined as
\begin{align}
    \widehat{\pi}(\m | \dobs) \propto \exp \left( -\| \OpCovNoise^{-1/2}(\paratoobs \m - \dobs) \|_{L^2([0, 5])}^2 - \|\mathcal{C}_{\rm{pr}}^{-1/2} \sum_{i=1}^{\dimParameter} \m_i v_i \|^2_{L^2(\Domain)} \right).
\end{align}
We refer to \cite{cui2016scalable} for a detailed introduction to this approximation.

\begin{table}[t]
    \centering
    \begin{tabular}{l|rr|rr}
         & initial guess & optimized path & change & change \\
         & path $1^*$ & path $1^+$ & (absolute) & (relative) \\
         \hline
        cost & 6.477e-03 & 5.879e-03 & -5.979e-04 & $\times$ 0.908 \\
        A-\ac{oed} & 6.359e-03 & 5.461e-03 & -8.981e-04 & $\times$ 0.859 
    \end{tabular}
    \caption{
    Cost function evaluation and A-\ac{oed} utility cost of the parameter-reduced posterior covariance matrix before and after path optimization.
    }
    \label{tab:highdim:optimization}
\end{table}

To account for the different scaling of the initial condition with the new high-dimensional prior compared to the $\dimParameter = 2$ examples in the previous section, we adjust the scaling of \eqref{ex:noise:pde:strong} in our noise model by factor 10 (i.e., the diffusion coefficient in \eqref{ex:noise:pde:strong} becomes 10, and the advection coefficient becomes 1000).
Furthermore, we adjust the scaling of the regularization to account for the reduction in prior uncertainty (the $\dimParameter=2$ example had $\ObjAOED(\CovPr) = 2$, the parameter-reduced prior has $\ObjAOED(\CovPr) = 0.029$).
We use path $1^*$ as initial guess for our optimization, and denote the obtained optimized path as $1^+$.
An overview on the optimization process is provided in Table \ref{tab:highdim:optimization}.
The optimization terminated successfully within 184 s (72 iterations, 81 objective function evaluations, 73 gradient evaluations, constraint violation of $4.735\cdot 10^{-10}$, NLP error of $3.029\cdot 10^{-9}$).
The A-\ac{oed} criterion was reduced from $6.359\cdot 10^{-3}$ for path $1^*$ to $5.461\cdot 10^{-3}$ for path $1^+$, a reduction of factor 0.859.
The smaller reduction in uncertainty compared to the previous experiments is likely because path $1^*$ was already optimized to distinguish information from the two main vortices in the domain (top-left and bottom-right).

The variance of each mode of the parameter-reduced prior and the parameter-reduced posterior for path $1^+$ is shown in Figure \ref{fig:highdim:energy} (right).
We see that primarily the uncertainty in the first, most dominant modes was reduced, while the variance in the later modes has changed little.
Yet, as expected, all eigenvalues of the parameter-reduced posterior covariance matrix $\CovPost$ are at least as small as those of the parameter-reduced prior.

To exemplify the influence of the path $\ObsPath$ on the solution of the Bayesian inverse problem, we take a sample of the prior (shown in Figure \ref{fig:highdim:sample}, top row), take measurements along path $1^+$, perturb them with noise drawn from our adjusted noise model, and compute the parameter-approximated posterior mean.
Its evolution in time is shown in Figure \ref{fig:highdim:sample} (second row).
Note that a perfect reconstruction is not expected; first, because the prior sample is in the full finite element space (dimension 32,096) while the parameter-approximated posterior mean is restricted to the reduced space $\mathcal{V}$ (dimension 25); and, second, because the posterior is by definition a compromise between the prior and the observational data.
Still, we observe a good agreement between the sample and the dynamics of the parameter-approximated posterior mean, especially when enough time has progressed for positive and negative pollutant sources to mix.

\section{Conclusion}\label{sec:conclusion}

In this work we introduced a linear Bayesian inverse problem for the identification of a parameter of interest from indirect, continuous-time measurements obtained by a moving sensor.
We posed the \ac{oed} problem of identifying the sensor path for which the posterior distribution is the most informative as measured by an \ac{oed} utility function.
The global optimization of the path as a whole allows the sensor to exploit transient dynamics in the state:
It can prioritize measurements for reducing the uncertainty in one parameter while waiting for information from another to come closer;
it can also accept temporarily taking uninformative measurements for the purpose of repositioning. 
We showed how to coherently discretize and constrain the optimization problem such that it converges in the time-continuous limit, derived formulas for the gradient of the cost function, and proposed how to consequently solve the optimization problem using an interior point method.
We demonstrated our proposed methodology on a convection-diffusion model of a pollutant spreading in an urban environment.

In this work we have limited our exposition to a single moving sensor but the extension to several sensors is straightforward.
With this extension, our methodology can moreover be applied for positioning stationary sensors---i.e., the sensor selection setting but with a non-discrete set of possible positions---that take continuous-time measurements.
In future work, we will increase the applicability of our approach through more general versions of the \ac{oed} utility functions (e.g., the expected information gain) to allow nonlinear parameter-to-observable maps.
We will expand upon our initial results for the parameter-reduced posterior in Section \ref{sec:results:transient:highdim} to target infinite-dimensional parameter spaces and further reduce computational costs through strategic model reduction and multifidelity techniques.
Moreover, we intend to increase the robustness of our \ac{oed} optimization problem to aleatoric uncertainties and investigate cheaper approaches for obtaining initial guesses.

\section{Acknowledgments}
This work was supported by the Applied Mathematics activity within the U.S. Department of Energy, Office of Science, Advanced Scientific Computing Research, under Contract DE-AC02-06CH11357.
NA and KW were supported in parts by the Department of Energy grant DE-SC002317, and the National Science Foundation grant \#2103942.

\section*{Code and data availability}
The source code and data used to generate the numerical results reported in Section \ref{sec:results:forward} are available at \href{https://github.com/leyffer/m2dtIceSheet}{https://github.com/leyffer/m2dtIceSheet}.

\appendix

\section{Proofs}\label{sec:proofs}

\begin{corollary}\label{thm:dinL2}
Let $u \in L^2([0, \finalt], \mathcal{V})$ with $\mathcal{V} \subset L^2(\Omega)$, $\ObsPath \in C([0, \finalt], \Omega)$, and let $\d$ be defined as in \eqref{eq:data:general}.
If \eqref{KernelAssumptions} holds, then $\d \in L^2([0, \finalt])$.
If additionally $u \in C([0, \finalt], \mathcal{V})$, then $\d \in C([0, \finalt])$.
\end{corollary}

\begin{proof}
The first part of the claim follows from the Cauchy inequality:
\begin{align*}
    \|\d\|_{L^2([0,\finalt])}^2 
    &= \int_0^{\finalt} \left( \int_\Domain \u(t, \mathbf{x}) \MeasKern(\mathbf{x}, \ObsPath(t)) d\mathbf{x} \right)^2 dt \\
    &\le \int_0^{\finalt} \|\u(t, \cdot)\|_{L^2(\Omega)}^2 \|\MeasKern(\cdot, \ObsPath(t))\|_{L^2(\Omega)}^2 dt \\
    &\le c_1 \left( \esssup_{\mathbf{y} \in \Omega} \|\MeasKern(\cdot, \mathbf{y})\|_{L^2(\Omega)} \right) \int_0^{\finalt} \|\u(t, \cdot)\|_{\mathcal{V}}^2 dt \\
    &\le c_2 \|\u\|_{L^2([0, T];\mathcal{V})}^2
\end{align*}
with constants $c_1, c_2 > 0$.
The second part follows from the linearity of $\d$ in $\u$.
\end{proof}

\begin{proof}[Proof of Proposition \ref{thm:continuity}]
    Let the time step index $1 \le k \le \dimTime$ be arbitrary.
    Let $\varepsilon > 0$ be arbitrary.
    Because $\pk \in \Domain$ and $\Domain$ is open, there exists $r > 0$ such that $B_r(\pk) \subset \Domain$ for all $1 \le k \le \dimTime$.
    Let $\mathbf{q} \in B_r(\pk)$ be arbitrary with $\|\pk-\mathbf{q}\|_2 < \varepsilon$.
    Define $\widetilde{\ObsPath} := \{\pk[1], \dots, \pk[k-1], \mathbf{q}, \pk[k+1], \dots, \pk[\dimTime]\}$.
    By definition \eqref{eq:paratoobs:discrete}, we have for $1 \le \ell \le \dimTime$, $\ell \neq k$, that $[\paratoobsDiscrete - \mathbf{G}(\widetilde{\ObsPath})]_{\ell,m} = 0$ for all $m \in \{1, \dots, \dimParameter\}$.
    For $\ell = k$ and any $1 \le m \le \dimParameter$ we have that
    \begin{align*}
        \left| [\paratoobsDiscrete - \mathbf{G}(\widetilde{\ObsPath})]_{k,m} \right|
        &= \left| \int_{\Domain} \pdesolutionoperatorDiscrete(\mbasis)(\tk, \mathbf{x}) \left( \MeasKern(\mathbf{x}, \pk) - \MeasKern(\mathbf{x}, \mathbf{q})\right) d\mathbf{x} \right| \\
        &\le \int_{\Domain} |\pdesolutionoperatorDiscrete(\mbasis)(\tk, \mathbf{x})| \left| \MeasKern(\mathbf{x}, \pk) - \MeasKern(\mathbf{x}, \mathbf{q})\right| d\mathbf{x} \\
        &\le \int_{\Domain} |\pdesolutionoperatorDiscrete(\mbasis)(\tk, \mathbf{x})| \gamma_{\MeasKern}(\mathbf{x}) \|\pk-\mathbf{q}\|_2 d\mathbf{x} \\
        &\le \overline{\gamma} \|\pk-\mathbf{q}\|_2 \| \pdesolutionoperatorDiscrete(\mbasis)(\tk, \,\cdot\,) \|_{L^1(\Domain)}.
    \end{align*}
    Since $\Domain$ is bounded, the Cauchy inequality yields that
    \begin{align*}
        \| \pdesolutionoperatorDiscrete(\mbasis)(\tk, \,\cdot\,) \|_{L^1(\Domain)} \le \sqrt{|\Domain|} \, \| \pdesolutionoperatorDiscrete(\mbasis)(\tk, \,\cdot\,) \|_{L^2(\Domain)}
    \end{align*}
    is bounded.
    In particular, this bound is independent of $\pk$ and $\mathbf{q}$.
    Thus, $\left| [\paratoobsDiscrete - \mathbf{G}(\widetilde{\ObsPath})]_{k,m} \right| \in \mathcal{O}(\varepsilon)$.
    The claim follows because $k$ and $m$ were arbitrary and because $\dimParameter$ and $\dimTime$ are finite.
\end{proof}

\begin{proof}[Proof of Proposition \ref{thm:gradient:1}]
 Let the time step index $1 \le k \le \dimTime$ and the spatial direction $i \in \{1,2,3\}$ (or $i \in \{1,2\}$ if $\Domain \subset \mathbb{R}^2$) be arbitrary.
 Let $\ObjOED \in \{\ObjAOED, \ObjDOED\}$.
In the following, we adopt the numerator (Jacobian) layout convention for matrix derivatives.

Using a standard scalar-by-matrix result from matrix calculus (\cite{abadir2005matrix},  p.354), we observe that
\begin{align}\label{eq:derivative:proof:1}
    \frac{d}{d \pki} \ObjOED(\CovPostDiscrete) = \operatorname{trace}\left( \left. \frac{d \ObjOED(X)}{d X}\right|_{X = \CovPostDiscrete} \frac{d}{d \pki} \CovPostDiscrete \right),
\end{align}
where $\frac{d \ObjAOED(X)}{d X} = \Identity[\dimParameter]$ for the A-\ac{oed} utility function $\ObjAOED(X) = \operatorname{trace}(X)$ (\cite{abadir2005matrix},  p.355), and $\frac{d \ObjDOED(X)}{d X} = \operatorname{det}(X) X^{-1}$ for the D-\ac{oed} utility function $\ObjDOED(X) = \operatorname{det}(X)$ (\cite{abadir2005matrix},  p.369).
Since $\CovPostDiscrete$ is symmetric postitive definite, we can expand further
\begin{align}\label{eq:derivative:proof:2}
    \frac{d}{d \pki} \CovPostDiscrete = - \CovPostDiscrete \left(\frac{d}{d \pki} \CovPostDiscrete^{-1} \right) \CovPostDiscrete
\end{align}
using the scalar-by-scalar identity $\frac{d}{d x}X(x)^{-1} = X(x)^{-1} \frac{d X}{d x} X(x)^{-1}$ for invertible matrices $X$ that depend on a scalar $x\in \mathbb{R}$ (\cite{abadir2005matrix},  p.364).
Inserted into \eqref{eq:derivative:proof:1}, \eqref{eq:derivative:proof:2} gives us
\begin{align}\label{eq:derivative:proof:A}
        \frac{d}{d \pki} \ObjAOED(\CovPostDiscrete) &= -\operatorname{trace}\left( \CovPostDiscrete \left(\frac{d}{d \pki} \CovPostDiscrete^{-1} \right) \CovPostDiscrete \right),
    \end{align}
for A-\ac{oed}, and 
\begin{align}\label{eq:derivative:proof:D}
    \frac{d}{d \pki} \ObjDOED(\CovPostDiscrete) &= -\ObjDOED(\CovPostDiscrete) \operatorname{trace}\left(\left(\frac{d}{d \pki} \CovPostDiscrete^{-1} \right) \CovPostDiscrete \right).
\end{align}
for D-\ac{oed}.

Because the path remains uninterrupted within $\Domain$, i.e. $\pk[\ell] \in \Omega$ for all $\ell \in \{1,\dots, \dimTime\}$, we can next exploit the structure of $\CovPostDiscrete$ in \eqref{eq:CovPost:discrete} to write
\begin{equation}\label{eq:derivative:proof:3}
    \begin{aligned}
    \frac{d}{d \pki} \CovPostDiscrete^{-1} &= \frac{d}{d \pki}\left(\paratoobsDiscrete^{\top}\CovNoise^{-1} \paratoobsDiscrete + \iCovPr\right) \\
    &= \left(\frac{d \paratoobsDiscrete}{d \pki}\right)^{\top} \CovNoise^{-1} \paratoobsDiscrete + \paratoobsDiscrete^{\top}\CovNoise^{-1}\frac{d \paratoobsDiscrete}{d \pki},
    \end{aligned}
\end{equation}
where we have used that the prior covariance matrix $\CovPr$ and and the inverse noise covariance matrix $\CovNoise^{-1}$ are independent of $\ObsPath$.
After inserting \eqref{eq:derivative:proof:3} back into \eqref{eq:derivative:proof:A} and \eqref{eq:derivative:proof:D}, we can simplify further by exploiting that $\frac{d}{d \pki} \CovPostDiscrete^{-1}$ only appears within the trace operator:
For the A-\ac{oed} utility cost function, we obtain \eqref{eq:derivative:A:discrete} through
\begin{align*}
    &\frac{d}{d \pki} \ObjAOED(\CovPostDiscrete) \\ 
    &= -\operatorname{trace}\left( \CovPostDiscrete \left(\frac{d}{d \pki} \CovPostDiscrete^{-1} \right) \CovPostDiscrete \right) \\
    &= -\operatorname{trace}\left( \CovPostDiscrete \left(\left(\frac{d \paratoobsDiscrete}{d \pki}\right)^{\top} \CovNoise^{-1} \paratoobsDiscrete + \paratoobsDiscrete^{\top}\CovNoise^{-1}\frac{d \paratoobsDiscrete}{d \pki} \right) \CovPostDiscrete \right) \\
    &= -2 \operatorname{trace} \left( \CovPostDiscrete  \paratoobsDiscrete^{\top}\CovNoise^{-1}\frac{d \paratoobsDiscrete}{d \pki} \CovPostDiscrete \right);
\end{align*}
for the D-\ac{oed} utility cost function, we get \eqref{eq:derivative:D:discrete} analogously through
\begin{align*}
    &\frac{d}{d \pki} \ObjDOED(\CovPostDiscrete) \\
    &= -\ObjDOED(\CovPostDiscrete) \operatorname{trace}\left(\left(\frac{d}{d \pki} \CovPostDiscrete^{-1} \right) \CovPostDiscrete \right) \\
    &= -\ObjDOED(\CovPostDiscrete) \operatorname{trace}\left(\left(\left(\frac{d \paratoobsDiscrete}{d \pki}\right)^{\top} \CovNoise^{-1} \paratoobsDiscrete + \paratoobsDiscrete^{\top}\CovNoise^{-1}\frac{d \paratoobsDiscrete}{d \pki} \right) \CovPostDiscrete \right) \\
    &= -\ObjDOED(\CovPostDiscrete) \operatorname{trace}\left( \paratoobsDiscrete^{\top}\CovNoise^{-1}\frac{d \paratoobsDiscrete}{d \pki} \CovPostDiscrete \right) \\
    &\qquad - \ObjDOED(\CovPostDiscrete) \operatorname{trace}\left(\left(\left(\frac{d \paratoobsDiscrete}{d \pki}\right)^{\top} \CovNoise^{-1} \paratoobsDiscrete \right) \CovPostDiscrete \right) \\
    &= -2\ObjDOED(\CovPostDiscrete) \operatorname{trace}\left( \paratoobsDiscrete^{\top}\CovNoise^{-1}\frac{d \paratoobsDiscrete}{d \pki} \CovPostDiscrete \right),
\end{align*}
where we have used in the last step that both $\CovPostDiscrete$ and $\paratoobsDiscrete^{\top}\CovNoise^{-1}\frac{d \paratoobsDiscrete}{d \pki}$ are square matrices and $\CovPostDiscrete$ is symmetric.

Continuity of $\nabla_{\pk} \ObjAOED(\CovPostDiscrete)$ and $\nabla_{\pk} \ObjDOED(\CovPostDiscrete)$ in $\pk$ follows from the representations \eqref{eq:derivative:A:discrete} and \eqref{eq:derivative:D:discrete}, as these are compositions of continuous operators.    
\end{proof}

\begin{proof}[Proof of Corollary \ref{thm:gradient:2}]
The matrix-by-scalar derivative $\frac{d \paratoobsDiscrete}{d \pki} \in \mathbb{R}^{\dimTime \times \dimParameter}$ can be evaluated element-wise using the definition of the discrete parameter-to-observable map in \eqref{eq:paratoobs:discrete}:
\begin{equation}\label{eq:derivative:proof:4}
    \begin{aligned}
    \left[\frac{d \paratoobsDiscrete}{d \pki} \right]_{\ell, m}
    &= \frac{d}{d \pki} \left[ \paratoobsDiscrete \right]_{\ell, m} \\
    &= \frac{d}{d \pki} \int_\Domain \uDiscrete_m(\tk[\ell], \mathbf{x}) \MeasKern(\mathbf{x}, \pk[\ell]) d\mathbf{x}
\end{aligned}
\end{equation}
with state $\uDiscrete_m := \pdesolutionoperatorDiscrete(\mbasis)$.
Note that $\left[\frac{d \paratoobsDiscrete}{d \pki} \right]_{\ell, m} = 0$ for $k \neq \ell$.
For $k=\ell$, we use the Leibniz integration rule on \eqref{eq:derivative:proof:4}, exploiting that the integration domain $\Domain$ is path-independent:
\begin{align}\label{eq:derivative:proof:5}
     \left[\frac{d \paratoobsDiscrete}{d \pki} \right]_{k, m} 
    &= \int_\Domain \uDiscrete_m(\tk[\ell], \mathbf{x}) \frac{d \MeasKern}{d \mathbf{y}_i} \left.(\mathbf{x}, \mathbf{y})
    \right|_{\mathbf{y}=\pki} d\mathbf{x}.
\end{align}
By our assumptions on $\MeasKern$, the integral \eqref{eq:derivative:proof:5} is well-defined and continuous in $\pki$.
\end{proof}

\begin{proof}[Proof of Corollary \ref{thm:gradient:3}]
Inserting the form \eqref{eq:measkern:alternative} into the definition of the parameter-to-observable matrix $\paratoobsDiscrete$ in \eqref{eq:paratoobs:discrete}, we get after a change of variables that
\begin{align*}
    [\paratoobsDiscrete]_{\ell,m} = \int_{B_r(0)} u^h_m(\tk, \mathbf{x} + \pk) \phi(\mathbf{x})d\mathbf{x}
\end{align*}
for any $1 \le \ell \le \dimTime$ and $1 \le m \le \dimParameter$.
The matrix-by-scalar derivative $\frac{d \paratoobsDiscrete}{d \pki} \in \mathbb{R}^{\dimTime \times \dimParameter}$ has thus the form
\begin{align}\label{eq:derivative:proof:6}
     \left[\frac{d \paratoobsDiscrete}{d \pki} \right]_{k, m} 
    &= \int_{B_r(0)} \frac{\partial u^h_m}{\partial \mathbf{x}_i}(\tk, \mathbf{x} + \pk) \phi(\mathbf{x})d\mathbf{x},
\end{align}
and is zero for all other rows $\ell \neq k$.
Note that \eqref{eq:derivative:proof:6} is bounded because $u_m^h(t_k) \in H^1(\Omega)$.
Continuity follows from the continuity of $\phi$ and the compactness of $\overline{B_r(\mathbf{0})}$.
\end{proof}

\bibliographystyle{elsarticle-num}
\bibliography{myreferences.bib}

\end{document}